\documentclass[12pt,letterpaper]{article}
\usepackage[margin=1in]{geometry}
\usepackage[T1]{fontenc}
\usepackage{newtxtext}
\usepackage{amsthm}

\usepackage{enumitem}
\setlist[itemize]{leftmargin=*}

\usepackage{resizegather}

\usepackage{amsmath,amssymb,amsfonts,mathrsfs}
\usepackage{graphicx}
\usepackage{mdframed}
\usepackage{textcomp}
\usepackage{xcolor}
\usepackage{textcomp}
\usepackage{tablefootnote}
\usepackage{threeparttable}
\usepackage{bbm}
\usepackage{dsfont}
\usepackage{tikz}
\usepackage{graphicx}
\usepackage{tkz-euclide}
\usepackage{algorithm}
\usepackage{algpseudocode}
\usepackage{booktabs}
\usetikzlibrary{positioning}
\usetikzlibrary{shapes}
\usetikzlibrary{shapes.misc}
\usetikzlibrary{shapes.geometric}
\usetikzlibrary{plotmarks}
\usetikzlibrary{intersections}
\usetikzlibrary{calc}
\usetikzlibrary{fit}
\usetikzlibrary{patterns,tikzmark}
\usetikzlibrary{matrix,decorations.pathreplacing,calc}

\tikzset{cross/.style={cross out, draw, 
         minimum size=2*(#1-\pgflinewidth), 
         inner sep=0pt, outer sep=0pt}}

 \usepackage{circuitikz} %
 
\usepackage{enumitem}
\setlist[itemize]{leftmargin=*}
\setlist[enumerate]{leftmargin=*}

\usepackage{booktabs,multirow}

\usepackage{hyperref}
\hypersetup{hidelinks}

\newcommand{\ci}[0]{u}

\newcommand{\x}[0]{x}

\newcommand{\vertiii}[1]{{\left\vert\kern-0.25ex\left\vert\kern-0.25ex\left\vert #1 
    \right\vert\kern-0.25ex\right\vert\kern-0.25ex\right\vert}}

\newtheorem{theorem}{Theorem}
\newtheorem{proposition}{Proposition}

\newtheorem{corollary}{Corollary}[theorem]
\newtheorem{lemma}{Lemma}

\newtheorem{problem}{Problem}
\newtheorem{assumption}{Assumption}
\newtheorem{remark}{Remark}

\newtheorem{definition}{Definition}

\newmdenv[
    linecolor=white, backgroundcolor=lightgray!15, innertopmargin=5pt, innerbottommargin=5pt, skipabove=10pt, skipbelow=10pt
]{graybox}

\makeatletter
\renewcommand{\fps@figure}{htp}
\renewcommand{\fps@table}{htp}
\makeatother

\newcommand{\mr}[1]{\mathrm{#1}}

\usepackage{comment}

\newtheorem{exmpl}{Example}

\newcommand{\parstart}[1]{\noindent \textbf{#1.}\;}

\usepackage{cite}
\usepackage{xparse}
\providecommand{\ijrncreadingcopy}{0}
\ifnum\ijrncreadingcopy=0
  \usepackage[tablesonly,nomarkers,nolists]{endfloat}
\fi
\newlength{\ijrnctablecaptionskip}
\makeatletter
\newcommand{\ijrncfigurecaptionsize}{\small}
\def\ijrnc@figuretype{figure}
\long\def\@makecaption#1#2{%
  \vskip\abovecaptionskip
  \begingroup
  \ifx\@captype\ijrnc@figuretype
    \ijrncfigurecaptionsize
    \sbox\@tempboxa{\textbf{#1:} #2}%
    \ifdim \wd\@tempboxa >\hsize
      \textbf{#1:} #2\par
    \else
      \global \@minipagefalse
      \hb@xt@\hsize{\hfil\box\@tempboxa\hfil}%
    \fi
  \else
    \sbox\@tempboxa{#1: #2}%
    \ifdim \wd\@tempboxa >\hsize
      #1: #2\par
    \else
      \global \@minipagefalse
      \hb@xt@\hsize{\hfil\box\@tempboxa\hfil}%
    \fi
  \fi
  \endgroup
  \vskip\belowcaptionskip}
\makeatother
\newenvironment{IEEEkeywords}{\par\smallskip\noindent\textbf{Keywords: }}{\par\medskip}
\newlength{\biophotowidth}
\newlength{\biogutter}
\newlength{\biotextwidth}
\NewDocumentEnvironment{IEEEbiography}{o m}{\par\bigskip\noindent
  \setlength{\biotextwidth}{\linewidth}%
  \IfValueT{#1}{\addtolength{\biotextwidth}{-\biophotowidth}%
    \addtolength{\biotextwidth}{-\biogutter}%
    \begin{minipage}[t]{\biophotowidth}\vspace{0pt}#1\end{minipage}\hfill}%
  \begin{minipage}[t]{\biotextwidth}\vspace{0pt}%
  \small
  \textbf{#2}\ \ignorespaces}{\par\end{minipage}\par\bigskip}
\newcommand{\IEEEPARstart}[2]{#1#2}

\hypersetup{hidelinks}
\newcommand{\ijrncrevision}[1]{#1}%
\definecolor{ink}{RGB}{20,20,20}
\definecolor{gfill}{RGB}{246,246,246}      %
\definecolor{gfill2}{RGB}{238,238,238}      %
\definecolor{linegray}{RGB}{95,95,95}       %
\definecolor{edgegray}{RGB}{70,70,70}       %


\title{Verification and Forward Invariance of Control Barrier Functions for Differential-Algebraic Systems
}

\author{Hongchao Zhang$^{1}$, Mohamad H. Kazma$^{2}$,\\
Meiyi Ma$^{1}$, Taylor T. Johnson$^{1}$ and Ahmad F. Taha$^{2}$
\thanks{$^{1}$ Department of Computer Science, $^{2}$ Department of Civil and Environmental
Engineering, Vanderbilt University, Nashville, TN 37235 USA.
{\tt\scriptsize \{hongchao.zhang, mohamad.h.kazma, meiyi.ma, taylor.johnson, ahmad.taha\}@vanderbilt.edu}
}
\thanks{This work was supported by the National Science Foundation under grants 2220401, 2443803, 2152450, and 2550934, and by DARPA under grant FA8750-23-C-0518. The authors' findings and conclusions do not necessarily represent the views of the funding agencies.}
}

\date{}
\begin{document}
\maketitle
\begin{center}\small\textbf{Running title:} Safety Control of DAE Systems\end{center}
\thispagestyle{empty}

\begin{abstract}
Differential-algebraic equations (DAEs) arise in power networks, chemical processes, and multibody systems, where algebraic constraints encode physical conservation laws. The safety of such systems is critical, yet safe control is challenging because algebraic constraints restrict allowable state trajectories. Control barrier functions (CBFs) provide computationally efficient safety filters for ordinary differential equation (ODE) systems. However, existing CBF methods are not directly applicable to DAEs due to potential conflicts between the CBF condition and the constraint manifold. This paper introduces DAE-aware CBFs that incorporate the differential-algebraic structure through projected vector fields. We derive conditions that ensure forward invariance of safe sets while preserving algebraic constraints and extend the framework to higher-index DAEs. A systematic verification framework is developed, establishing necessary and sufficient conditions for geometric correctness and feasibility of DAE-aware CBFs. For polynomial systems, sum-of-squares certificates are provided, while for nonpolynomial and neural network candidates, satisfiability modulo theories are used for falsification. The approach is validated on wind turbine and flexible-link manipulator systems.
\end{abstract}

\begin{IEEEkeywords}
	Differential-algebraic systems, control barrier functions, safety-critical control, verification, cyber-physical systems.
\end{IEEEkeywords}

\section{Introduction}\label{sec:intro}

\IEEEPARstart{D}{YNAMIC}  systems are ubiquitous in engineering and science, and ordinary differential equations (ODEs) represent one of the most common modeling frameworks. Safety of dynamical systems is typically formulated as the positive invariance of a safe set~\cite{blanchini1999set}, requiring that trajectories starting within a prescribed region remain there for all time. Approaches to ensuring safety range from offline certification methods such as Hamilton-Jacobi reachability~\cite{bansal2017hamilton} and learning-enabled approaches~\cite{brunke2022safe}, to online safe control. For online safe control, \textit{control barrier functions} (CBFs) have emerged as a computationally efficient tool~\cite{ames2014control,ames2019control}. CBFs synthesize safety filters that take a desired control input from any performance-based policy and, if the input is unsafe, minimally modify it in real-time to guarantee safety. This framework has seen numerous extensions, including high-order CBFs for systems with high relative degree~\cite{xiao2021high} and neural network representations~\cite{dawson2023safe,zhang2024seev}.

\begin{figure}[t]
    \centering
    \includegraphics[scale=0.8]{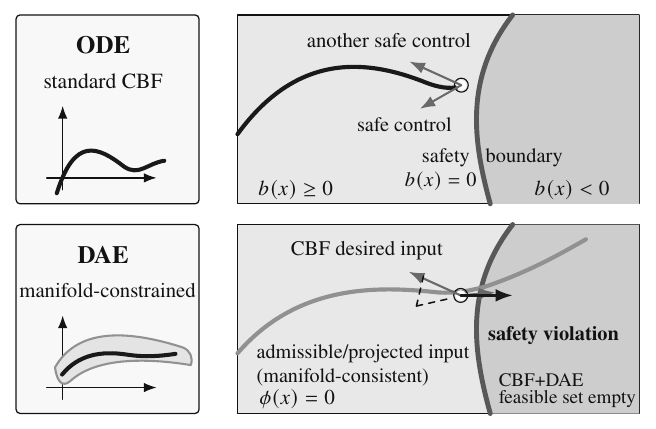}
    \caption{Conceptual illustration of hidden infeasibility in standard CBF safety filters when applied to DAE systems. The algebraic constraint manifold $\mathcal{M} = \{\x : \phi(\x) = 0\}$ imposes an additional compatibility requirement on the control input that is not accounted for by standard CBF conditions, which may render the safety filter infeasible even when the state lies within the safe set $\mathcal{C}$.}
    \label{fig:dae_cbf_conceptual}
\end{figure}

In engineering applications including power networks~\cite{Grob2016,Kazma2025e,albustami2026powercbfs}, chemical processes~\cite{Volker2005}, and multibody mechanical systems~\cite{Park2024b}, systems are naturally modeled as differential-algebraic equations (DAEs) because they inherently involve algebraic constraints, often stemming from physical laws like conservation of mass and energy. 
In ODE systems, CBFs have shown great promise in online safe control, yet are not directly applicable to DAE systems due to the algebraic constraints that restrict allowable state trajectories. As visualized in Fig. \ref{fig:dae_cbf_conceptual}, a fundamental assumption underlying CBF safety is the existence of a feasible control input satisfying the barrier condition, which may be violated due to the algebraic constraints (see Fig.~\ref{fig:dae_cbf_conceptual}). Recently, reference~\cite{Park2024b} introduced CBFs for flexible-link manipulators modeled as DAEs, but a general framework remains an open problem.

\parstart{Related Work} CBFs have emerged as a computationally efficient solution, enforcing set invariance through quadratic program (QP) formulations that serve as real-time safety filters~\cite{ames2014control,ames2019control}. Recent work has explored neural network representations of CBFs to improve expressiveness and scalability, with advances in verification~\cite{abate2021fossil,zhang2023exact,edwards2024fossil} and synthesis~\cite{dawson2022safe,dawson2023safe,zhang2024seev,tayal2025cp} enabling formal safety guarantees for neural CBFs. This line of work also extends to stochastic systems \cite{tayal2024learning,zhang2025stochastic,wooding2025protect}. While these methods are well-established for ODEs, their extension to DAEs presents fundamental challenges due to the evolution of the differential states along the constraint manifold defined by the algebraic constraints.

The safety problem for DAEs remains significantly unexplored. Reference~\cite{Dang2004} applied hybrid-systems reachability techniques to analog and mixed-signal circuits modeled as DAEs by integrating an equivalent ODE at each step and numerically projecting the result onto the constraint manifold, thereby introducing projection-error accumulation. Reference~\cite{Cross2008} extended Hamilton–Jacobi level-set reachability to nonlinear index-1 DAEs by deriving two schemes: an analytic projection of the dynamics onto the constraint manifold and a full-state formulation using closest-point extensions. Although these methods demonstrate that classical reachability ideas can be adapted to DAE settings, they address only offline safety verification. The verification for DAE systems has been explored primarily through reachability-based approaches in the formal methods community, including benchmarks with simulation and discrepancy functions~\cite{Musau2018}, star-set propagation for large-scale linear DAEs~\cite{Tran2019}, and scalable zonotope-based algorithms for nonlinear index-1 DAEs~\cite{Althoff2014}. Barrier certificates have also been applied to DAE safety verification using sum-of-squares (SOS) programming~\cite{Pedersen2016}. However, these frameworks are fundamentally diagnostic: they certify whether a pre-defined controller maintains safety but do not provide constructive methods to synthesize feedback laws that enforce safety in real-time under differential-algebraic constraints.

The control synthesis for DAEs has targeted asymptotic stability rather than hard constraint satisfaction. Existing approaches construct polynomial Lyapunov functions via SOS programming~\cite{Kundu2015} or establish dissipativity conditions using integral quadratic constraints and LMI-based verification~\cite{Jensen2025}. Yet, asymptotic stability is insufficient for safety-critical applications; a system can be stable while transiently violating critical state constraints such as voltage limits or physical boundaries. Consequently, the rich literature on DAE stabilization does not directly translate to the strict set-invariance requirements of safety-critical control.

Joint DAE-aware safety-filter design and verification of correctness and feasibility were identified as an open problem in~\cite{kazma2026safe}. Standard geometric control approaches for manifold-constrained systems typically assume explicit constraint representations, rather than implicitly defined algebraic equations. CBFs for DAEs were first introduced in~\cite{Park2024b}, enabling safety control of flexible-link manipulators without converting the model to an ODE. However, their approach assumes that trajectories remain in the consistency space without explicitly deriving how control inputs enforce algebraic constraint manifold invariance. Here, we construct projected vector fields that encode the differential-algebraic structure and make the role of compatibility conditions explicit, including how algebraic state derivatives are determined. In addition,~\cite{Park2024b} does not formalize the relationship between the differentiation index of the DAE and the relative degree of the barrier function---a relationship that is critical for understanding when CBF constraints remain feasible in the presence of hidden algebraic constraints. Moreover, while~\cite{Park2024b} provides theoretical safety guarantees, it does not develop computational methods for formally verifying candidate CBFs. Our framework applies to higher-index DAEs, explicitly characterizes how the differentiation index impacts CBF design, and provides formal verification tools based on SOS programming with Satisfiability Modulo Theories (SMT) based falsification for non-polynomial cases.

\parstart{Paper Contributions} In this paper, we address these aforementioned limitations by introducing a novel framework for safety-critical control for DAE systems. As illustrated in Fig.~\ref{fig:dae_cbf_conceptual}, applying standard CBF-based safety filters to DAE systems can lead to infeasibilities and, consequently, safety violations. This failure stems from two fundamental challenges: (C1) how to ensure that the control input maintains invariance of the constraint manifold while simultaneously enforcing set invariance, and (C2)  how to verify that the safety filter remains feasible when the differential equations are coupled with the algebraic constraints. Our framework addresses both challenges by developing DAE-aware CBF conditions and formal verification methods. The contributions are as follows.
\begin{itemize}
    \item 
    We present a class of CBFs, termed DAE-aware CBFs, that explicitly account for the differential-algebraic structure of the system through projected vector fields along the constraint manifold. We derive DAE-aware CBF conditions for the safety of nonlinear DAE systems while ensuring forward invariance of safe sets and consistency of the algebraic constraints.
    \item We generalize the DAE-aware CBF framework to handle higher-index DAE systems. We formalize the relationship between the differentiation index and the relative degree, showing that both notions must be accounted for to ensure safety; this leads to high-order CBF conditions that account for hidden constraints within the DAE structure. 
    \item We develop verification methods that certify both geometric correctness (the barrier set is contained in the safe set on the constraint manifold) and feasibility (existence of compatible control inputs satisfying both safety and algebraic constraints). For CBFs represented by polynomials, we provide SOS certificates. For CBFs represented by feed-forward neural networks, we formulate nonlinear programs and employ SMT for falsification. 
    \item We validate the proposed framework through simulations on benchmark DAE systems, demonstrating that DAE-aware CBFs maintain invariance of the constraint manifold and ensure feasibility of the safety filter. As a result, we numerically show that safety is preserved without violating the algebraic constraints, in contrast to standard CBF approaches, which become infeasible.
\end{itemize}

\parstart{Paper Organization}The remainder of this work is organized as follows. Section~\ref{sec:problem} provides preliminaries on DAE models and control barrier functions, and presents the formal problem formulation. Section~\ref{sec:dae_cbf} introduces DAE-aware CBFs and derives safety conditions for general DAE systems. Section~\ref{section:verification} develops computational methods for verifying the validity of a candidate DAE-aware CBF. Finally, Section~\ref{sec:conclusion} concludes the paper and outlines future research directions.
\section{Problem Formulation}\label{sec:problem}
In this section, we introduce the system model representation, safety specification, and formulate the problem.
\subsection{System Model and Preliminaries}
\label{subsec:model}

We consider a general continuous-time nonlinear control affine DAE system. The states of the system are described by dynamic states denoted as $\x_{d}(t) \in \mathbb{R}^{n_{d}}$ and algebraic states denoted as $\x_{a}(t) \in \mathbb{R}^{n_{a}}$. The overall state vector is given by $\x(t) = [\x_{d}(t)^{T}, \x_{a}(t)^{T}]^{T} \in \mathbb{R}^{n_{x}}$, where $n_{x} = n_{d} + n_{a}$. The control input is denoted as $\ci(t) \in \mathbb{R}^{n_{u}}$. The system dynamics can be described by the following general semi-explicit nonlinear DAE form
\begin{subequations}
\label{eq:dyn}
\begin{align}
    \dot{\x}_{d}(t) &= f(\x_{d}(t), \x_{a}(t)) + g(\x_{d}(t),\x_{a}(t)) \ci(t), \label{eq:dyn_diff} \\
    0 &= \phi(\x_{d}(t), \x_{a}(t)), \label{eq:dyn_alg}
\end{align}
\end{subequations}
where $f: \mathbb{R}^{n_x} \rightarrow \mathbb{R}^{n_d}$, $g: \mathbb{R}^{n_x} \rightarrow \mathbb{R}^{n_d} \times \mathbb{R}^{n_u}$, and $\phi: \mathbb{R}^{n_d} \times \mathbb{R}^{n_a} \rightarrow \mathbb{R}^{n_m}$ are smooth functions. Note that~\eqref{eq:dyn} can represent a broad class of systems, including mechanical systems satisfying holonomic or nonholonomic constraints, power network dynamics satisfying power flow equations, chemical reactors, water quality distribution systems that are modeled through partial differential equations.
\begin{assumption}
    \label{assumption:no_input}
    Throughout this paper, we assume that the algebraic constraint $\phi(\x)$ does not depend on the input $\ci$. 
\end{assumption}

While DAEs offer a modeling framework for several applications, numerical challenges often arise because the system dynamics evolve on a constraint manifold $\mathcal{M} = \{\x \in \mathbb{R}^{n_x}: \phi(\x) = 0\}$ that may be defined implicitly~\cite{Linh2009}. This manifold, determined by the algebraic constraints, is hidden in the sense that its geometric structure and dimension are not immediately apparent from the constraint equations alone. Additional constraints can be revealed by performing repeated differentiations of the algebraic constraints. This concept is formalized through the notion of a differentiation index. 
\begin{definition}[Differentiation index]\label{def:inde(1)}
	The differentiation index~\cite{Volker2005} of a descriptor system, denoted as $\nu$, refers to the number of differentiations required to obtain an ODE using algebraic manipulations.
\end{definition}

The generalization of the concept of differentiation index is the strangeness index~\cite[Definition 4.4]{Volker2005}. DAEs of index one are considered strangeness-free if no additional differentiation of the algebraic constraints is required to reveal hidden constraints. Note that DAEs of index greater than one can be transformed into an equivalent strangeness-free form through index reduction techniques~\cite{Linh2009,Volker2005}.
Furthermore, regularity is an important property for DAEs; it is a condition for the existence of a consistent unique solution for every consistent initial condition~\cite{Grob2016}. 
In practice, regularity, which implies the solvability of the DAE system, is a prerequisite for the solvability of physical systems modeled as DAEs.

\begin{definition}[Regularity]
    \label{def:regular}
	Regularity of a DAE~\eqref{eq:dyn} around an initial state can be characterized by matrix pair $(E, A)$, such that it is regular if and only if $\mr{det}(sE - A) \neq 0 $ for $s \in \mathbb{C}$.
\end{definition}
Here $E$ is the singular mass matrix and $A$ is the state space matrix obtained from linearizing the DAE system around the initial state. The above definition is standard in the literature on DAE systems~\cite{Volker2005,Grob2016}; it implies that a DAE is regular if the matrix pencil $(sE - A)$ is regular, i.e., the determinant is not identically zero. That being said, it ensures the DAE is well-posed: by the theory of matrix pencils~\cite{Volker2005}, a regular DAE has a unique solution for every consistent initial condition, meaning that the system dynamics are uniquely determined once the algebraic constraints are satisfied.

\subsection{Safety and Control Barrier Functions}
\label{subsec:safety}
Safety of dynamic systems usually refers to the property where the system remains inside of a pre-defined set of states. Let $\mathcal{C}:= \{\x\in\mathbb{R}^{n_x}: h(\x) \geq 0\}\subseteq \mathcal{X}$ be a closed set denoted as the safe set, for some function $h: \mathcal{X} \rightarrow \mathbb{R}$. Let $\partial\mathcal{C}$ denote the boundary 
of the set $\mathcal{C}$. The system is safe if the system remains in the safe set $\mathcal{C}$ for all time. This property is also referred to as positive invariance and is defined as follows.

\begin{definition}[Safety]
    \label{def:positive-invariance}
    A set $\mathcal{C} \subseteq \mathbb{R}^{n_x}$ is positive invariant under dynamics (\ref{eq:dyn}) and control policy $\mu$ if $\x(0) \in \mathcal{C}$ and $\ci(t) = \mu(\x(t)) \ \forall t \geq 0$ implies that $\x(t) \in \mathcal{C}$ for all $t \geq 0$. 
\end{definition}

One approach towards designing safe control policies for ODE systems is to use CBF. A CBF is a continuously differentiable function 
$b(\x): \mathbb{R}^{n_x} \rightarrow \mathbb{R}$ whose 0-level set defines 
$\mathcal{D} := \{\x \in \mathbb{R}^{n_x} : b(\x) \geq 0\}\subseteq \mathcal{C}$. Let $\alpha: \mathbb{R} \rightarrow \mathbb{R}$ denote a class-$\mathcal{K}_{\infty}$ function (i.e., strictly increasing and satisfies $\alpha(0) = 0$). The following result guarantees positive invariance of the set $\mathcal{D}$ under appropriate control.

\begin{theorem}[\hspace{-0.01cm}\cite{ames2019control}]
\label{theorem:CBF-safety}
Suppose that $b(\x)$ is a CBF such that $b(\x(0)) \geq 0$, and the input $\ci(t)$ satisfies
\begin{equation}
\label{eq:CBF-def}
\frac{\partial b}{\partial \x}\big(f(\x) + g(\x)\ci\big) \geq -\alpha\big(b(\x)\big)
\end{equation}
for all $t \geq 0$. Then the set $\mathcal{D} = \{\x \in \mathbb{R}^{n_x} : b(\x) \geq 0\}$ is positive invariant.
\end{theorem}

Theorem~\ref{theorem:CBF-safety} implies that if $b(\x)$ is a CBF, then enforcing the constraint~\eqref{eq:CBF-def} on the control input $\ci(t)$ at each time step is sufficient to guarantee safety. In practice, this is achieved by solving a CBF quadratic program (\textbf{CBF-QP}) at each time step that minimizes deviation from a nominal controller $\ci_{\mathrm{nom}}$ while satisfying the CBF constraint:
\begin{equation}
\label{eq:CBF-QP}
\begin{aligned}
    & \min_{\ci(t) \in \mathcal{U}}  \|\ci(t) - \ci_{\mathrm{nom}(t)}\|^2 
    \\
    & \text{s.t.} \quad  \frac{\partial b}{\partial \x}\big(f(\x(t)) + g(\x(t))\ci(t)\big) \geq -\alpha\big(b(\x(t))\big).
\end{aligned}
\end{equation}
We note that Theorem~\ref{theorem:CBF-safety} assumes the existence of a control input that satisfies the CBF condition~\eqref{eq:CBF-def}. This motivates the following definition.
\begin{definition}[Valid CBF]
\label{def:valid-cbf}
A function $b(\x): \mathbb{R}^{n_x} \rightarrow \mathbb{R}$ is a \emph{valid control barrier function} for the safe set $\mathcal{C}$ under dynamics~\eqref{eq:dyn} if the following conditions hold
\begin{enumerate}
    \item (Correctness) $\mathcal{D}\subseteq \mathcal{C}$, i.e., $b(\x) \geq 0$ implies $h(\x) \geq 0$ for all $\x \in \mathbb{R}^{n_x}$. This ensures that the CBF-defined set is contained within the true safe set.
    \item (Feasibility) For all $\x \in \mathcal{D}$, there exists a control input $\ci \in \mathbb{R}^{n_u}$ such that $\frac{\partial b}{\partial \x}\big(f(\x) + g(\x)\ci\big) \geq -\alpha\big(b(\x)\big)$.
\end{enumerate}
\end{definition}
We denote a state $\x_{ce}^{(c)} \in \partial \mathcal{D} \subseteq \mathbb{R}^{n_x}$ that violates \eqref{eq:Nagumo-CPI} as a \emph{safety counterexample}.

\subsection{Problem Formulation: DAE-Aware CBFs}
For DAE systems~\eqref{eq:dyn}, the system trajectories are restricted to the constraint manifold $\mathcal{M}$ defined by the algebraic constraints. Following the CBF approach from Section~\ref{subsec:safety}, we seek to construct a barrier function $b(\x)$ that defines a set $\mathcal{D}\subseteq \mathcal{C}$. The problem of designing safe control policies for DAE systems is formalized as follows. 

\begin{problem}
\label{problem:dae_cbf}
Given a DAE system described by \eqref{eq:dyn} and a safe set $\mathcal{C}$, design a safe control policy $\mu: \mathbb{R}^{n_x} \rightarrow \mathbb{R}^{n_u}$ such that the DAE system \eqref{eq:dyn} maintains positive invariance of $\mathcal{D} \cap \mathcal{M}$ for all initial states $\x(0) \in \mathcal{D} \cap \mathcal{M}$. 
\end{problem}
\begin{figure}[t]
  \centering
  \includegraphics[scale=0.67]{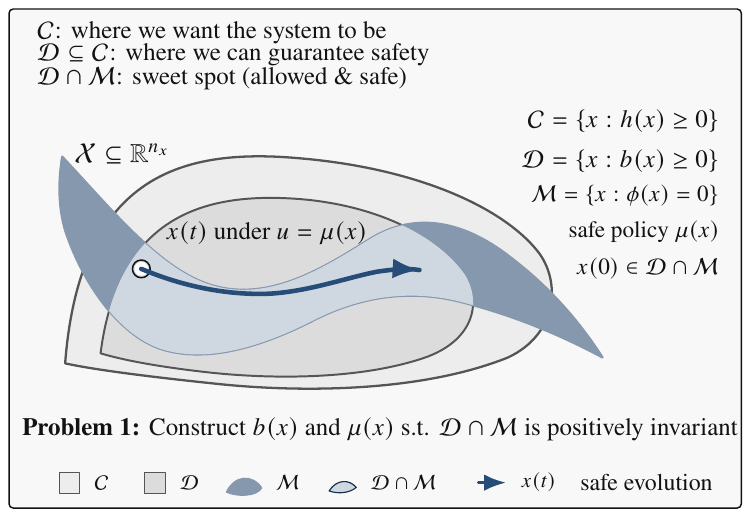}
  \caption{Geometric illustration of Problem~\ref{problem:dae_cbf}. The state space contains the constraint manifold $\mathcal{M}$, the safe set $\mathcal{C}$, and the candidate barrier set $\mathcal{D}$. The goal is to design a barrier function $b(\x)$ and a safe control policy $\mu$ such that $\mathcal{D} \cap \mathcal{M} \subseteq \mathcal{C} \cap \mathcal{M}$ is positively invariant under the closed-loop DAE dynamics~\eqref{eq:dyn}, ensuring that all system trajectories initialized in $\mathcal{D} \cap \mathcal{M}$ evolve on the constraint manifold and remain within the safe set for all $t \ge 0$.}
  \label{fig:problem1_concept}
\end{figure}

While the CBF approach from Section~\ref{subsec:safety} is well-established for ODE systems, directly applying standard CBFs to DAE systems presents fundamental challenges. For DAE systems, the existence of a control input satisfying the CBF condition~\eqref{eq:CBF-def} is not guaranteed because the standard CBF formulation does not account for the algebraic constraints $\phi(\x) = 0$. Specifically, a control input that satisfies the CBF inequality $\frac{\partial b}{\partial \x}\big(f(\x) + g(\x)\ci\big) \geq -\alpha\big(b(\x)\big)$ may violate the constraint-manifold compatibility condition, leading to infeasibility of the CBF-QP~\eqref{eq:CBF-QP} on $\mathcal{M}$ and resulting in safety violations. 
The following example illustrates the limitations of applying traditional CBF methods to DAE systems.

\begin{exmpl}[Motivating Example]\label{exmpl:DAE_unaware_example_1}
\begin{figure*}[t]
    \centering
    \includegraphics[width=0.95\linewidth]{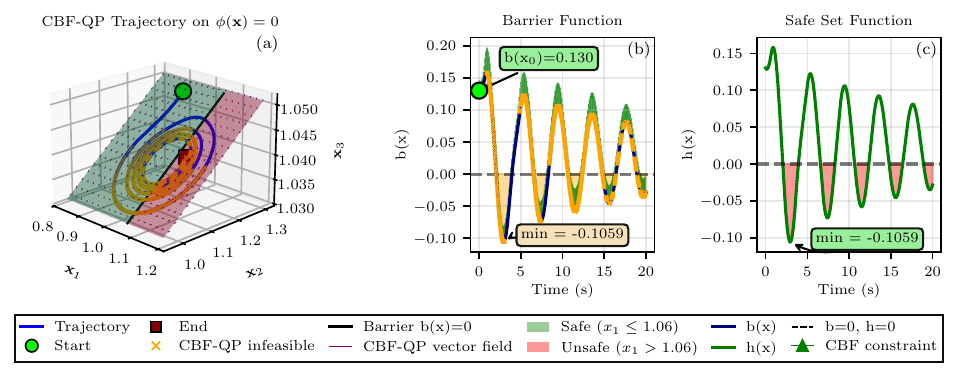}
    \caption{Simulation results for Example~\ref{exmpl:DAE_unaware_example_1} using a DAE-unaware CBF. The figure shows (a) the CBF-QP trajectory on the constraint manifold $\phi(\x)=0$, (b) the evolution of the barrier function $b(\x)$ (with $b(\x_0)=0.130$ and a minimum $ b(\x)=-0.1059$), and (c) the evolution of the safe set function $h(\x)$ (a minimum $h(\x)=-0.1059$). The CBF-QP becomes infeasible when the CBF inequality conflicts with the algebraic equality constraints, leading to violation of the safety condition $b(\x)\ge 0$.}
    \label{fig:motivating_infeasible}
\end{figure*}
Consider a nonlinear DAE model derived from a wind turbine power system~\cite{Tsourakis2009}:
\begin{subequations}\label{eq:wind_turbine}
    \begin{align}
    & \dot{x}_{1}=\alpha_1\left(u - \beta_1 x_3(x_2 - x_3)\right), \\
    & \dot{x}_{2}=\alpha_2\left(x_{1}-x_3\right), \\
    & 0=x_3^4 - \left[\beta_2 + \beta_3 x_3(x_2 - x_3)\right] x_3^2 + \beta_4,
    \end{align}
\end{subequations}
where $\x_d=(x_1, x_2)$ are the differential states, $\x_a=x_3$ is the algebraic state, $\alpha_i$ and $\beta_j$ are system parameters, and $\ci$ is the control input. The safe set is defined as $\mathcal{C} = \{\x \in \mathbb{R}^{3} : h(\x) = x_1 - x_{\max} \geq 0\}$. A traditional CBF $b(\x) = h(\x)$ is constructed without accounting for the algebraic constraint. Simulation results using a standard CBF QP are shown in Fig.~\ref{fig:motivating_infeasible}. The trajectory (Fig.~\ref{fig:motivating_infeasible}a) evolves on the constraint manifold until the QP becomes infeasible. This infeasibility causes the system to enter the intersection of the unsafe region and the constraint manifold. As shown in Fig.~\ref{fig:motivating_infeasible}b), the safety function $h(\x)$ drops below zero after infeasibility occurs. The CBF QP becomes infeasible because the control input required to render the safe set forward invariant may violate the algebraic constraints. Consequently, the associated QP admits an empty feasible set.
\end{exmpl}

\section{DAE-Aware CBFs}
\label{sec:dae_cbf}

This section introduces the main theoretical results; specifically, the framework for constructing CBFs that explicitly account for the algebraic constraints inherent in DAE systems. Tab.~\ref{tab:dae_cbf_roadmap} summarizes the main results of this section.

\begin{table}[h]
\centering
\setlength{\belowcaptionskip}{\ijrnctablecaptionskip}
\caption{Summary of Main Results for DAE-Aware CBF Construction, outlining the key concepts in each subsection.}
\label{tab:dae_cbf_roadmap}
\begin{tabular}{@{}lll@{}}
\toprule
\textbf{Subsection} & \textbf{Key Concept} & \textbf{Main Results} \\ \midrule
\ref{subsec:prelim} Preliminaries & Lie derivatives + relative deg. & Defs. \ref{def:lie_derivative}-\ref{def:classical_relative_degree}; Thm. \ref{theorem:HOCBF-safety} \\
\ref{subsec:projected_fields} Projected Fields & Revealing hidden dynamics & Eq. \ref{eq:projected_fields}; Lem. \ref{lemma:index-relative-degree} \\
\ref{subsec:index1} Index-1 DAEs & Compatibility + safety & Thm. \ref{thm:positive_invariance}; Cor. \ref{lemma:dae_hocbf_index1} \\
\ref{subsec:index2plus} Higher-Index & $\nu$-order diff. + projection & Prop. \ref{prop:manifold_invariance_indexnu}; Thm. \ref{thm:positive_invariance_indexnu} \\ \bottomrule
\end{tabular}
\end{table}

\subsection{Preliminaries}\label{subsec:prelim}
We first introduce some preliminaries on Lie derivatives and relative degree for control-affine systems, which will be used in the subsequent development of DAE-aware CBFs.

For a control-affine system $\dot{\x} = f(\x) + g(\x)\ci$, the Lie derivative can be defined according to the following definition.
\begin{definition}[Lie derivative~\cite{khalil2002nonlinear}]
\label{def:lie_derivative}
The \emph{Lie derivative} of a smooth scalar function $h(\x)$ along a smooth vector field $f(\x)$ is the directional derivative of $h(\x)$ in the direction of $f(\x)$, defined as
\begin{equation*}
\mathcal{L}_f h(\x) := \frac{\partial h}{\partial \x} f(\x).
\end{equation*}
Furthermore, the Lie derivative, for $k \geq 1$, can be defined recursively as $\mathcal{L}_f^k h(\x) := \mathcal{L}_f\left(\mathcal{L}_f^{k-1} h(\x)\right)$ with $\mathcal{L}_f^0 h(\x) := h(\x)$.
\end{definition}
For a matrix-valued function $g(\x)$, the Lie derivative generalizes to $\mathcal{L}_g h(\x) := \frac{\partial h}{\partial \x} g(\x)$. The classical notion of relative degree for a scalar function with respect to a control-affine system is defined as follows.

\begin{definition}[Relative degree~\cite{Isidori2013}]
\label{def:classical_relative_degree}
A smooth function $h(\x)$ has relative degree $d \geq 1$ with respect to the system $\dot{\x} = f(\x) + g(\x) \ci$ if $d$ is the smallest integer such that 
\begin{enumerate}
    \item $\mathcal{L}_g  \mathcal{L}_f^k h(\x)=0$ for all $\x$ in a neighborhood of $\x^{\circ}$ and all $k<d-1$;
    \item $ \mathcal{L}_g  \mathcal{L}_f^{d-1} h\left(\x\right) \neq 0$.
\end{enumerate}
\end{definition}

Note that, by considering the relative degree, Theorem~\ref{theorem:CBF-safety} applies when the barrier function has relative degree one, meaning that the control input appears directly in the first time derivative of $b$. 
For barrier functions with an arbitrary relative degree $r \geq 1$, the following result introduces the higher-order CBF framework.

\begin{theorem}[High-Order CBF (HOCBF)~{\cite[Theorem~4]{xiao2021high}}]
\label{theorem:HOCBF-safety}
Consider the ODE system $\dot{\x} = f(\x) + g(\x)\ci$ and suppose $b(\x)$ has relative degree $r \geq 1$. Let the HOCBF hierarchy 
be defined as
\begin{equation*}
\psi_0(\x) := b(\x), \qquad \psi_i(\x) :=  \mathcal{L}_f \psi_{i-1}(\x) + \alpha_i(\psi_{i-1}(\x)), \qquad i = 1, \ldots, r-1,
\end{equation*}
then the set $\mathcal{D} = \{\x : b(\x) \geq 0\}$ is positively invariant under any initial condition $\x(0) \in \mathcal{D}$ with $\psi_i(\x(0)) \geq 0$ for $i = 0, \ldots, r-1$ if there exists a control $\ci$ such that
\begin{equation}
\label{eq:hocbf_condition}
 \mathcal{L}_f \psi_{r-1}(\x) + \alpha_{r-1}(\psi_{r-1}(\x)) +  \mathcal{L}_g \psi_{r-1}(\x) \, \ci \geq 0
\end{equation}
holds for all $\x \in \mathcal{D}$.
\end{theorem}

\parstart{Revisiting Example~\ref{exmpl:DAE_unaware_example_1}} we revisit the wind turbine DAE~\eqref{eq:wind_turbine} and analyze safety using the standard HOCBF framework of Theorem~\ref{theorem:HOCBF-safety}. On the constraint manifold $\mathcal{M}$, the algebraic constraint implicitly determines $x_3$ as a smooth function of $(x_1, x_2)$. Consequently, the control input $\ci$ steers the differential states $\x_d = (x_1, x_2)$ through~\eqref{eq:dyn_diff}, which in turn constrains the algebraic state $x_3$ via the implicit relation $\phi(\x_d, x_3) = 0$. This $\ci \to \x_d \to \x_a$ projection chain is the mechanism by which safety enforcement on $\x_d$ hold for DAE system \eqref{eq:dyn} on manifold $\mathcal{M}$.

The following theorem provides necessary and sufficient conditions for a set to be positive invariant.

\begin{theorem}[Nagumo's Theorem~{\cite[Section~4.2]{blanchini2008set}}]
\label{theorem:Nagumo}
A closed set $\mathcal{D}$ is \emph{controlled positively invariant} if, for every $\x \in \partial \mathcal{D}$, where $\partial \mathcal{D}$ denotes the boundary of $\mathcal{D}$, there exists a control input $\ci \in \mathcal{U}$ such that
\begin{equation}
\label{eq:Nagumo-CPI}
f(\x) + g(\x)\ci \in \mathcal{A}_{\mathcal{D}}(\x),
\end{equation}
where $\mathcal{A}_{\mathcal{D}}(\x)$ denotes the tangent cone to $\mathcal{D}$ at $\x$.
\end{theorem}

Theorem~\ref{theorem:Nagumo} states that a set $\mathcal{D}$ is positive invariant if and only if for every state on the boundary of $\mathcal{D}$, there exists a control input such that the derivative of the states lies within the tangent cone of $\mathcal{D}$. This condition ensures that the system cannot leave the safe set $\mathcal{D}$ once it is inside. The next section presents preliminary results that will be used to derive the verification conditions for safety CBFs.
\subsection{Projected Vector Fields}\label{subsec:projected_fields}
We begin by characterizing the relationship between the dynamic state $\x_{d}$ and the algebraic state $\x_{a}$ imposed by the constraint manifold $\mathcal{M} = \{\x : \phi(\x) = 0\}$, leading to the definition of \emph{projected vector fields} that capture the constrained dynamics on the manifold defined by the algebraic equations. For any feasible trajectory of the DAE system \eqref{eq:dyn}, the algebraic constraints must remain satisfied at all times, i.e., $\phi(\x(t)) = 0$ for all $t \geq 0$. This requirement creates a hierarchical structure: the control input $\ci$ directly influences the dynamic state $\x_{d}$ through the differential equation $\dot{\x}_{d} = f(\x) + g(\x)\ci$, and the dynamic state in turn affects the algebraic state $\x_{a}$ through the implicit constraint $\phi(\x_{d}, \x_{a}) = 0$. This indirect coupling is fundamental to understanding how safety can be enforced in DAE systems. 

To analyze how the algebraic constraints couple the differential and algebraic states, we decompose the Jacobian of $\phi$ with respect to the dynamic variables $\x_{d}$ and the algebraic variables $\x_{a}$ as
\begin{align*}
    J_{d}(\x) &:= \nabla_{\x_{d}} \phi(\x), & J_{a}(\x) &:= \nabla_{\x_{a}} \phi(\x).
\end{align*}
We also define $f_{d}(\x) := f(\x)$ and $g_{d}(\x) := g(\x)$ as the drift and input matrices that appear in \eqref{eq:dyn}. Differentiating the constraint $\phi(\x(t)) = 0$ with respect to time yields the following \emph{compatibility condition}
\begin{equation}
\label{eq:dae_compatibility}
    J_{d}(\x) \big(f_{d}(\x) + g_{d}(\x) \ci\big) + J_{a}(\x) \dot{\x}_{a} = 0.
\end{equation}
This condition reveals that the algebraic state derivative $\dot{\x}_{a}$ is not directly controlled but must be implicitly determined to maintain consistency with the algebraic constraints. The solvability of \eqref{eq:dae_compatibility} for $\dot{\x}_{a}$ depends critically on the rank properties of $J_{a}(\x)$, which in turn determines the \emph{differentiation index} of the DAE system.

To construct a DAE-aware CBF, we must understand how the control input $\ci$ propagates through this hierarchical structure. The key question is: how many differentiation steps are required before $\ci$ affects a given function of the state? This depends on two distinct concepts: the \emph{differentiation index} of the algebraic constraint and the \emph{relative degree} of the barrier function on the manifold $\mathcal{M}$.

For DAE systems, we extend this notion to the algebraic constraint $\phi(\x)$ to characterize how many times each component must be differentiated before the control appears explicitly. 
The following fundamental result connects the differentiation index $\nu$ to the relative degree $d$.
Since the algebraic constraint $\phi$ couples to only a subset of the differential states through $J_d^{(\nu)}$, the standard relative degree of $\x_d$ may overcount the actuation available to $\phi$. We therefore define the \emph{constraint-coupled relative degree} $d'$ of the differential subsystem \eqref{eq:dyn_diff} as the smallest integer such that $J_d^{(\nu)}(\x)\, \mathcal{L}_g \mathcal{L}_f^{d^{\prime}-1} \x_d(\x) \neq 0$.

\begin{lemma}
\label{lemma:index-relative-degree}
Consider the DAE system \eqref{eq:dyn} with differentiation index $\nu$ and constraint-coupled relative degree $d'$. Then the relative degree $d$ of the algebraic constraint $\phi$ with respect to the full DAE system satisfies
\begin{equation*}
d = \nu + d^{\prime} - 1.
\end{equation*}
In particular, when $d^{\prime} = 1$ (i.e., $J_d^{(\nu)} g_d \neq 0$), we have $d = \nu$.
\end{lemma}
\begin{proof}
After $\nu$ differentiations of $\phi$, the algebraic state $\x_a$ becomes expressible as a function of $\x_d$ (Definition~\ref{def:inde(1)}), and the $\nu$-th derivative $\phi^{(\nu)}$ depends on $\x_d$ through $J_d^{(\nu)}$. For $\ci$ to appear in $\phi^{(\nu)}$, we need $d'$ additional derivatives of $\x_d$ as seen through $J_d^{(\nu)}$ (Definition~\ref{def:classical_relative_degree}). Hence $d = \nu + d' - 1$.
\end{proof}

This lemma reveals that the relative degree of the algebraic constraint decomposes into contributions from the DAE index and the actuation structure of the differential states. The result generalizes the relationship discovered in~\cite[Lemma 1]{di2019globally} to arbitrary relative degrees of the differential subsystem.

\parstart{Revisiting Example~\ref{exmpl:DAE_unaware_example_1}}
For the wind turbine DAE~\eqref{eq:wind_turbine}, $J_a(\x) = \partial\phi/\partial x_3 \neq 0$ generically, so $\nu = 1$. Since $J_d = [0,\; -\beta_3 x_3^3]$ and $g_d = [\alpha_1,\; 0]^\top$, we have $J_d g_d = 0$; that is, $\ci$ does not appear after a single differentiation of $\phi$. However, $\dot{x}_2 = \alpha_2(x_1 - x_3)$ depends on $x_1$, whose derivative contains $\ci$, so $J_d \mathcal{L}_g \mathcal{L}_f \x_d \neq 0$ and $d' = 2$. By Lemma~\ref{lemma:index-relative-degree}, $d = 1 + 2 - 1 = 2$, reflecting the two-hop chain $\ci \to x_1 \to x_2 \to \phi$.

For simplicity, let $\eta^{(\nu)}(\x)$ denote the \emph{control influence matrix} at level $\nu$, defined as
\begin{equation}
    \label{eq:eta}
\eta^{(\nu)}(\x) := (J_d^{(\nu)}(\x) f_d(\x))^{d^{\prime}-1} J_d^{(\nu)}(\x).
\end{equation}
When $d^{\prime} = 1$, this simplifies to $\eta^{(\nu)}(\x) = J_d^{(\nu)}(\x)$. Let $J_{ext}(\x) = \begin{bmatrix} J_a^{(\nu)}(\x) & \eta^{(\nu)}(\x) g_d(\x) \end{bmatrix}$; we make the following assumption for the remainder of the paper.
\begin{assumption}
\label{assumption:regularity}
We assume that the DAE system~\eqref{eq:dyn} has Jacobian $J_{ext}(\x)$ with full row rank for all $x \in \mathcal{C} \cap \mathcal{M}$.
\end{assumption}

The above assumption is not restrictive for physical systems that are commonly modeled as DAEs, since such systems naturally admit strangeness-free representations (power system generator and electric circuit models~\cite{Kazma2025e,Duan2016}) or can be transformed into such forms (mechanical system models~\cite{Abdalmoaty2021,Park2024b}) using index reduction techniques~\cite{Volker2005,Linh2009}. 
Assumption~\ref{assumption:regularity} requires the extended Jacobian $J_{ext}(\x)$ to have full row rank for all $\x \in \mathcal{M}$. This allows for systems where $J_a^{(\nu)}(\x)$ alone may be rank deficient, but the control input $\ci$, acting through $\eta^{(\nu)}(\x) g_d(\x)$, provides additional degrees of freedom to satisfy the compatibility condition. When $J_a^{(\nu)}(\x)$ has full row rank, the extended Jacobian trivially has full row rank and the projection $P^{(\nu)}(\x) = 0$.

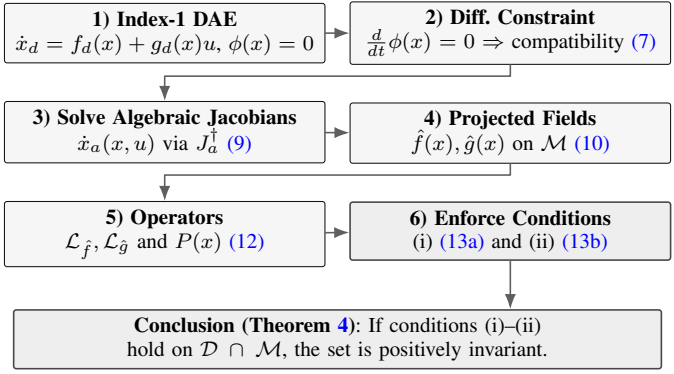
\begin{figure}[t]
    \centering
\begin{tikzpicture}[
    >=Latex,
    font=\rmfamily\fontsize{8pt}{9pt}\selectfont, 
    line cap=round,
    line join=round,
    node distance=5mm and 3.2mm, %
    box/.style={
      draw=ink,
      rounded corners=1pt,
      line width=0.4pt,
      fill=gfill,
      align=center,
      inner sep=3pt,
      text width=#1
    },
    emphbox/.style={
      draw=edgegray,
      rounded corners=1pt,
      line width=0.6pt,
      fill=gfill2,
      align=center,
      inner sep=3pt,
      text width=#1
    },
    arr/.style={->, draw=linegray, line width=0.6pt}
]

\def\colW{4.05cm}

\node[box=\colW] (A) {%
\textbf{1) Index-1 DAE}\\
$\dot x_d=f_d(x)+g_d(x)u$, $\phi(x)=0$
};

\node[box=\colW, right=of A] (B) {%
\textbf{2) Diff. Constraint}\\
$\frac{d}{dt}\phi(x)=0 \Rightarrow$ compatibility \eqref{eq:dae_compatibility}
};

\node[box=\colW, below=of A] (C) {%
\textbf{3) Solve Algebraic Jacobians}\\
$\dot x_a(x,u)$ via $J_a^\dagger$ \eqref{eq:dae_min_norm}
};

\node[box=\colW, below=of B] (D) {%
\textbf{4) Projected Fields}\\
$\hat f(x),\hat g(x)$ on $\mathcal{M}$ \eqref{eq:projected_fields}
};

\node[box=\colW, below=of C] (E) {%
\textbf{5) Operators}\\
$\mathcal L_{\hat f}, \mathcal L_{\hat g}$ and $P(x)$ \eqref{eq:projection_operator}
};

\node[emphbox=\colW, below=of D] (F) {%
\textbf{6) Enforce Conditions}\\
(i) \eqref{eq:invariance_projection} and (ii) \eqref{eq:invariance_cbf}
};

\node[emphbox=8.4cm, below=6mm of $(E.south)!0.5!(F.south)$, anchor=north] (G) {%
\textbf{Conclusion (Theorem \ref{thm:positive_invariance})}: If conditions (i)--(ii) hold on $\mathcal{D}\cap\mathcal{M}$, the set is positively invariant.
};

\draw[arr] (A) -- (B);
\draw[arr] (B.south) -- ++(0,-0.15) -| (C.north);
\draw[arr] (C) -- (D);
\draw[arr] (D.south) -- ++(0,-0.15) -| (E.north);
\draw[arr] (E) -- (F);

\draw[arr] (F.south) -- (F.south |- G.north);

\end{tikzpicture}
    \caption{Roadmap for the derivation of DAE-aware CBF conditions for index-1 systems. The algebraic constraints $\phi(\x) = 0$ are differentiated to expose the compatibility condition and the constraint Jacobians $J_d(\x)$ and $J_a(\x)$. Under the invertibility assumption on $J_a(\x)$, the algebraic states are eliminated via the implicit function theorem, and the DAE~\eqref{eq:dyn} is projected onto the constraint manifold $\mathcal{M}$. The CBF condition is then lifted back to the full state space, producing DAE-aware CBF conditions that jointly enforce manifold compatibility and positive invariance of $\mathcal{D} \cap \mathcal{M}$.}
    \label{fig:roadmapCBF}
\end{figure}

\subsection{CBFs for Index-1 DAEs}\label{subsec:index1}
We first consider the class of index-1 DAE systems in~\eqref{eq:dyn}.  
The derivation of the proposed DAE-aware CBF conditions is summarized in Fig.~\ref{fig:roadmapCBF}. Starting from the differential–algebraic dynamics, the algebraic constraint $\phi(\x)=0$ is differentiated to expose the compatibility condition and the associated constraint Jacobians. This enables the elimination of the algebraic state and the construction of projected vector fields on the constraint manifold $\mathcal{M}$. The resulting projected dynamics allow the CBF condition to be formulated consistently with the algebraic constraints.

Under Assumption~\ref{assumption:regularity}, any trajectory starting from $\x(0) \in \mathcal{M}$ remains on $\mathcal{M}$ for all $t \geq 0$ if and only if the control input $\ci(t)$ satisfies the compatibility condition \eqref{eq:dae_compatibility}. This structural property ensures that the compatibility condition \eqref{eq:dae_compatibility} can always be satisfied by an appropriate choice of $(\dot{\x}_a, \ci)$. To see this, note that manifold invariance requires $\phi(\x(t)) = 0$ for all $t \geq 0$. Differentiating with respect to time and applying the chain rule yields
\begin{equation*}
\frac{d}{dt}\phi(\x(t)) = J_d(\x) \dot{\x}_d + J_a(\x) \dot{\x}_a = 0.
\end{equation*}

Substituting the differential dynamics $\dot{\x}_d = f_d(\x) + g_d(\x)\ci$ and
transposing gives the compatibility condition \eqref{eq:dae_compatibility}.
Rewriting this condition shows that the algebraic state derivative must satisfy
\begin{equation*}
J_a(\x)\dot{\x}_a = -J_d(\x)\big(f_d(\x) + g_d(\x)\ci\big).
\end{equation*}
Since the extended Jacobian $J_{ext}(\x)$ has full row rank by Assumption~\ref{assumption:regularity}, this linear system is consistent and admits a solution for $\dot{\x}_a$. A solution is given by
\begin{equation}
\label{eq:dae_min_norm}
\dot{\x}_a(\x,\ci)
    = -J_a(\x)^{\dagger} J_d(\x)\big(f_d(\x) + g_d(\x)\ci\big),
\end{equation}
where $J_a(\x)^{\dagger}$ denotes the Moore-Penrose 
pseudoinverse. Substituting this expression into $J_a(\x)\dot{\x}_a$ verifies that
$\frac{d}{dt}\phi(\x(t))=0$, ensuring that the trajectory remains on the constraint manifold. 

This solution provides a natural construction for the constrained dynamics on $\mathcal{M}$ in control-affine form. We define the \emph{projected vector fields} by augmenting the dynamic state derivatives with the corresponding algebraic state derivatives from~\eqref{eq:dae_min_norm} as
\begin{equation}
\label{eq:projected_fields}
\begin{aligned}
    \hat{f}(\x) &:=
        \begin{bmatrix}
            f_{d}(\x) \\
            -J_{a}(\x)^{\dagger} J_{d}(\x) f_{d}(\x)
        \end{bmatrix},\qquad
    \hat{g}(\x) &:=
        \begin{bmatrix}
            g_{d}(\x) \\
            -J_{a}(\x)^{\dagger} J_{d}(\x) g_{d}(\x)
        \end{bmatrix}.
\end{aligned}
\end{equation}

Here, $\hat{f}(\x)$ represents the closed-loop dynamics when $\ci = 0$, while $\hat{g}(\x)$ captures the input-to-state map on $\mathcal{M}$. By construction, both fields automatically satisfy the compatibility condition \eqref{eq:dae_compatibility}, and the constrained dynamics on $\mathcal{M}$ admit the control-affine form $\dot{\x} = \hat{f}(\x) + \hat{g}(\x) \ci$. To verify that these fields are tangent to $\mathcal{M}$, note that 
\begin{align*}
\nabla\phi(\x) \hat{f}(\x)
&= J_d(\x) f_d(\x)
   - J_a(\x) J_a(\x)^\dagger J_d(\x) f_d(\x), \\
\nabla\phi(\x) \hat{g}(\x)
&= J_d(\x) g_d(\x)
   - J_a(\x) J_a(\x)^\dagger J_d(\x) g_d(\x).
\end{align*}
Since $J_a(\x)$ has full row rank, we have 
$J_a(\x) J_a(\x)^\dagger J_d(\x) = J_d(\x)$, and therefore both expressions
equal zero. Hence any trajectory evolving under 
$\dot{\x} = \hat{f}(\x) + \hat{g}(\x)\ci$ remains on~$\mathcal{M}$.

For any differentiable scalar function $\psi: \mathbb{R}^{n_x} \to \mathbb{R}$ (not to be confused with the algebraic constraint $\phi$), we write the Lie derivatives along the projected vector fields as
\begin{equation}
    \mathcal{L}_{\hat{f}} \psi(\x) := \nabla \psi(\x)^{\top} \hat{f}(\x),
    \quad
    \mathcal{L}_{\hat{g}} \psi(\x) := \nabla \psi(\x)^{\top} \hat{g}(\x),
\end{equation}
and define higher-order Lie derivatives recursively, e.g., $\mathcal{L}_{\hat{f}}^{2} \psi = \mathcal{L}_{\hat{f}}(\mathcal{L}_{\hat{f}} \psi)$ and $\mathcal{L}_{\hat{g}} \mathcal{L}_{\hat{f}} \psi = \mathcal{L}_{\hat{g}}(\mathcal{L}_{\hat{f}} \psi)$. When $\psi(\x) = \bar{\psi}(\x_{d})$ depends only on the dynamic state, we use the shorthand $L_{f_{d}} \bar{\psi}$ and $L_{g_{d}} \bar{\psi}$ for $\nabla_{\x_{d}} \bar{\psi}(\x_{d})^{\top} f_{d}(\x)$ and $\nabla_{\x_{d}} \bar{\psi}(\x_{d})^{\top} g_{d}(\x)$, respectively.

We define the projection operator
\begin{equation}
\label{eq:projection_operator}
P(\x) := I_{n_m} - J_a(\x) J_a(\x)^{\dagger},
\end{equation}
The compatibility condition requires $P(\x) \big(J_{d}(\x) f_{d}(\x) + \eta^{(1)}(\x) g_{d}(\x) \ci\big) = 0$, which constrains the control input $\ci$ to ensure manifold invariance. For index-1 DAEs, $\eta^{(1)}(\x) = J_d(\x)$ (defined in~\eqref{eq:eta} with $\nu = 1$ and $d' = 1$).

\begin{remark}\label{remark:projection_operator}
The projection operator $P(\x)$ defined in \eqref{eq:projection_operator} projects onto the null space of $J_a(\x)^\top$. When $J_a(\x)$ has full row rank, $P(\x) = 0$ and the compatibility condition is trivially satisfied. When $J_a(\x)$ is rank deficient, $P(\x) \neq 0$ and the condition $P(\x) (J_d(\x) f_d(\x) + \eta^{(1)}(\x) g_d(\x)\ci) = 0$ imposes a constraint on $\ci$. For higher-index DAEs, one must differentiate the algebraic constraints multiple times before applying projection techniques.
\end{remark}

Based on Assumption~\ref{assumption:regularity} and Remark~\ref{remark:projection_operator}, the dynamics
$$
\dot{\x} = \hat{f}(\x) + \hat{g}(\x)\ci,
$$
are well-defined and locally Lipschitz on $\mathcal{C} \cap \mathcal{M}$, allowing the proposed DAE-aware CBF conditions to be applied on the manifold.
The following theorem is the main result of this subsection. It provides sufficient conditions for the positive invariance of $\mathcal{D} \cap \mathcal{M}$. 

\begin{theorem}[Safety of Index-1 DAEs]
\label{thm:positive_invariance}
Consider the DAE system~\eqref{eq:dyn} with differentiation index $\nu = 1$ satisfying Assumption~\ref{assumption:no_input} and~\ref{assumption:regularity}. Suppose that $b(\x)$ is a CBF such that $b(\x(0)) \geq 0$ and $\x(0) \in \mathcal{D} \cap \mathcal{M}$. The set $\mathcal{D} \cap \mathcal{M}$ is positively invariant if the input $\ci(t)$ satisfies, for all $\x \in \mathcal{D} \cap \mathcal{M}$,
\begin{subequations}
\label{eq:invariance_conditions}
\begin{align}
P(\x) \big(J_{d}(\x) f_{d}(\x) + \eta^{(1)}(\x) g_{d}(\x) \ci\big) &= 0, \label{eq:invariance_projection} \\
\mathcal{L}_{\hat{f}} b(\x) + \mathcal{L}_{\hat{g}} b(\x) \ci &\geq -\alpha_{0}(b(\x)), \label{eq:invariance_cbf}
\end{align}
\end{subequations}
where $\eta^{(1)}(\x)$ is defined in~\eqref{eq:eta} (which simplifies to $J_d$ for $\nu = 1$).
\end{theorem}

\begin{proof}
By Assumption~\ref{assumption:no_input}, $\phi$ is independent of $\ci$, so differentiating $\phi(\x(t))=0$ and substituting $\dot{\x}_{d}=f_{d}(\x)+g_{d}(\x)\ci$ yields the compatibility condition~\eqref{eq:dae_compatibility}. Condition~\eqref{eq:invariance_projection} is equivalent to~\eqref{eq:dae_compatibility}. Since $\nu=1$, the Jacobian $J_{a}(\x)$ has full row rank on $\mathcal{C}\cap\mathcal{M}$, so $\dot{\x}_{a}$ is uniquely determined by~\eqref{eq:dae_min_norm}. By Assumption~\ref{assumption:regularity}, $J_{ext}(\x)$ has full row rank, so the projected fields~\eqref{eq:projected_fields} are well-defined and satisfy $\nabla\phi(\x)\hat{f}(\x)=0$, $\nabla\phi(\x)\hat{g}(\x)=0$. Hence any trajectory of $\dot{\x}=\hat{f}(\x)+\hat{g}(\x)\ci$ starting on $\mathcal{M}$ remains on $\mathcal{M}$. On $\mathcal{M}$ the dynamics reduce to the ODE $\dot{\x}=\hat{f}(\x)+\hat{g}(\x)\ci$, so
\begin{equation*}
\dot{b}(\x(t)) = \mathcal{L}_{\hat{f}} b(\x) + \mathcal{L}_{\hat{g}} b(\x)\,\ci.
\end{equation*}
Condition~\eqref{eq:invariance_cbf} yields $\dot{b}\geq-\alpha_{0}(b)$. Since $\alpha_{0}\in\mathcal{K}_{\infty}$ with $\alpha_{0}(0)=0$ and $b(\x(0))\geq 0$, the comparison lemma~\cite[Lemma~3.4]{khalil2002nonlinear} gives $b(\x(t))\geq 0$ for all $t\geq 0$. Combining manifold invariance and barrier nonnegativity, any trajectory starting in $\mathcal{D}\cap\mathcal{M}$ remains in $\mathcal{D}\cap\mathcal{M}$ for all $t\geq 0$.
\end{proof}

The key insight of Theorem~\ref{thm:positive_invariance} is that the projected vector fields reduce the DAE to an equivalent ODE on $\mathcal{M}$, after which the standard CBF condition of Theorem~\ref{theorem:CBF-safety} applies directly: condition~\eqref{eq:invariance_cbf} is exactly~\eqref{eq:CBF-def} with $(f, g)$ replaced by $(\hat{f}, \hat{g})$, while condition~\eqref{eq:invariance_projection} ensures that trajectories remain on $\mathcal{M}$. This is precisely the framework in Fig.~\ref{fig:roadmapCBF}.

\parstart{Revisiting Example~\ref{exmpl:DAE_unaware_example_1}}
The contrast between the DAE-unaware and DAE-aware approaches is illustrated by comparing Fig.~\ref{fig:motivating_infeasible} and Fig.~\ref{fig:synth_results_daecbf}. In Fig.~\ref{fig:motivating_infeasible}, the standard CBF-QP of Theorem~\ref{theorem:CBF-safety} is applied directly to the DAE system without accounting for the algebraic constraints; the resulting QP becomes infeasible and safety is violated. In Fig.~\ref{fig:synth_results_daecbf}, the DAE-aware CBF-QP of Theorem~\ref{thm:positive_invariance} jointly enforces the compatibility condition~\eqref{eq:invariance_projection} and the barrier inequality~\eqref{eq:invariance_cbf} along the projected dynamics. The trajectory remains on $\mathcal{M}$, and both $b(\x)$ and $h(\x)$ stay positive throughout, confirming that $\mathcal{D} \cap \mathcal{M}$ is positively invariant as guaranteed by the theorem.

In contrast to the relative degree of the algebraic constraint, we must also consider the \emph{relative degree of the barrier function} $b$ on the constraint manifold $\mathcal{M}$. This concept determines how many times $b$ must be differentiated along the constrained dynamics before the control appears. The following corollary generalizes Theorem~\ref{thm:positive_invariance} to the high-order CBF framework for index-1 DAE systems. 

\begin{corollary}[DAE-aware HOCBF for index-1 systems]
\label{lemma:dae_hocbf_index1}
Consider an index-1 DAE system satisfying Assumption~\ref{assumption:regularity}. Let $b(\x)$ have relative degree $d \geq 1$ on $\mathcal{M}$ with respect to the projected dynamics \eqref{eq:projected_fields}. Define the HOCBF hierarchy 
\begin{align*}
\psi_0(\x) := b(\x), \qquad 
\psi_i(\x) := \mathcal{L}_{\hat{f}} \psi_{i-1}(\x) + \alpha_i\big(\psi_{i-1}(\x)\big),
\end{align*}
for $i = 1, \ldots, d-1$ and where $\alpha_i$ are class-$\mathcal{K}_\infty$ functions. Then $\mathcal{D} \cap \mathcal{M}$ is positively invariant if there exists $\ci$ such that
\begin{subequations} 
\label{eq:dae_hocbf_conditions}
\begin{align}
P(\x) \big(J_d(\x) f_d(\x) + \eta^{(1)}(\x) g_d(\x)\ci\big) &= 0, \label{eq:dae_hocbf_compat} \\
\mathcal{L}_{\hat{f}} \psi_{d-1}(\x) + \alpha_{d-1}(\psi_{d-1}(\x)) + \mathcal{L}_{\hat{g}} \psi_{d-1}(\x) \, \ci &\geq 0, \label{eq:dae_hocbf_cbf}
\end{align}
\end{subequations}
hold for all $\x \in \mathcal{D} \cap \mathcal{M}$.
\end{corollary}
\begin{proof}
Condition \eqref{eq:dae_hocbf_compat} ensures $\x(t) \in \mathcal{M}$ by Theorem~\ref{thm:positive_invariance}. Condition \eqref{eq:dae_hocbf_cbf} is Theorem~\ref{theorem:HOCBF-safety} applied to the projected ODE $\dot{\x} = \hat{f}(\x) + \hat{g}(\x)\ci$ on $\mathcal{M}$, ensuring $b(\x(t)) \geq 0$.
\end{proof}

Corollary~\ref{lemma:dae_hocbf_index1} establishes that the projected dynamics framework extends naturally to barrier functions of any relative degree $d \geq 1$ on $\mathcal{M}$. The conditions~\eqref{eq:dae_hocbf_conditions} decompose into (\emph{i}) condition~\eqref{eq:dae_hocbf_compat}, which enforces manifold invariance through a \emph{single} compatibility equation, and (\emph{ii}) condition~\eqref{eq:dae_hocbf_cbf}, which enforces safety via the HOCBF hierarchy built from the projected Lie derivatives $\mathcal{L}_{\hat{f}}$ and $\mathcal{L}_{\hat{g}}$. 

When $d = 1$, the hierarchy reduces to $\psi_0 = b$ and condition~\eqref{eq:dae_hocbf_cbf} reduces exactly to the projected CBF inequality~\eqref{eq:invariance_cbf} of Theorem~\ref{thm:positive_invariance}. For higher-index DAEs ($\nu \geq 2$), the above conditions requires additional considerations; the details are provided in the subsequent section.

\subsection{CBFs for Higher-Index DAEs}
\label{subsec:index2plus}
For index-$\nu$ DAEs with $\nu \geq 2$, a single differentiation of $\phi$ cannot uniquely resolve $\dot{x}_a$. The control input $u$ appears only after $\nu$ successive differentiations of $\phi$, each revealing a new hidden algebraic constraint and yielding a corresponding compatibility condition. This produces $\nu$ compatibility conditions and higher-order projection operators $P^{(k)}(x)$, from which the DAE-aware CBF conditions are derived; the derivation roadmap is illustrated in Fig.~\ref{fig:flowchartHID}. The following remark characterizes the structural property of the extended constraint system that follows from Assumption~\ref{assumption:regularity}.

\begin{remark}[Index reduction and full row rank]
\label{remark:index_reduction}
For index-$\nu$ DAE systems with $\nu \geq 2$, Assumption~\ref{assumption:regularity} implies that after index reduction the extended constraint system yields a strangeness-free representation. Specifically, the algebraic constraint $\phi(x)=0$ is differentiated $\nu$ times and all derivatives $\dot{x},\ddot{x},\ldots$ are eliminated using the differential equation $\dot{x}_d = f_d(x)+g_d(x)u$. The resulting extended constraints $\phi^{(k)}(x)=0$ for $k=0,\ldots,\nu$ are functions of $x$ only, and the extended Jacobian $J_{ext}^{(\nu)}(x) = \begin{bmatrix} J_a^{(\nu)}(x) & \eta^{(\nu)}(x) g_d(x) \end{bmatrix}$ has full row rank for all $x \in \mathcal{M}$. This yields an equivalent strangeness-free representation on which the projected dynamics and CBF conditions are defined.
\end{remark}

Thus, $J_a^{(\nu)}$ and $J_d^{(\nu)}$ in \eqref{eq:top_level_jacobians} are well-defined only after the index-reduced extended constraint system is expressed as functions of the state $x$ alone.  This is the standard approach in high-index DAE analysis and control~\cite{Volker2005,Linh2009}.

This structural property ensures the projection operator $P(\x)$ is well-defined for the extended constraint system. The control input $\ci$ appears only after $\nu$ differentiations of the algebraic constraints, leading naturally to HOCBFs on $\mathcal{M}$. Manifold invariance requires that $\phi(\x(t)) = 0$ for all $t \ge t_0$. For an index-$\nu$ DAE, this is equivalent to requiring that all derivatives of $\phi$ up to order $\nu-1$ vanish along the trajectory. The first condition can be written as
\begin{equation}
\label{eq:phi_dot_index2}
\frac{d}{dt} \phi(\x(t)) = J_d(\x) \dot{\x}_d + J_a(\x) \dot{\x}_a = 0.
\end{equation}

For index-1 DAEs, full row rank of $J_a(\x)$ guarantees that \eqref{eq:phi_dot_index2} uniquely determines $\dot{\x}_a$ given $\dot{\x}_d = f_d(\x) + g_d(\x) \ci$. For index-$\nu$ DAEs with $\nu \geq 2$, rank deficiency of $J_a(\x)$ renders \eqref{eq:phi_dot_index2} insufficient, necessitating higher-order derivatives.

Successive differentiation produces a sequence of compatibility conditions. The second derivative, given by
\begin{equation*}
\frac{d^2}{dt^2} \phi(\x(t)) = \frac{d}{dt}\big(J_d(\x) \dot{\x}_d + J_a(\x) \dot{\x}_a\big) = 0,
\end{equation*}
involves $\ddot{\x}_d$ and $\ddot{\x}_a$. Continuing to the $\nu$-th derivative yields
\begin{equation}
\label{eq:phi_nu_derivative}
\frac{d^\nu}{dt^\nu} \phi(\x(t)) = 0.
\end{equation}
By Lemma~\ref{lemma:index-relative-degree}, $\ci$ appears explicitly with nonzero coefficient in \eqref{eq:phi_nu_derivative}, defining the differentiation index $\nu$ as the minimum number of differentiations required for control to appear.

The cascade of constraints $\frac{d^k \phi}{dt^k} = 0$ for $k = 1,\ldots,\nu$ provides a hierarchical set of conditions that $\ci$ must satisfy for manifold invariance. While index-1 DAEs require a single compatibility condition, index-$\nu$ DAEs require $\ci$ to simultaneously satisfy $\nu$ conditions; however, only the highest-order condition \eqref{eq:phi_nu_derivative} directly involves $\ci$ with a nonzero coefficient, while the lower-order conditions serve to ensure consistency of the trajectory with the algebraic constraints.

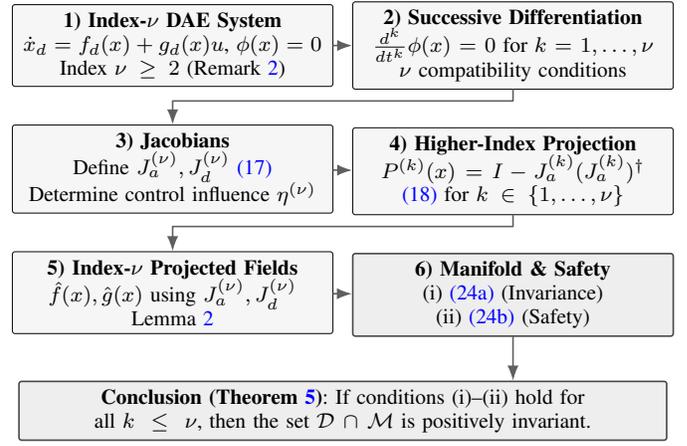
\begin{figure}[t]
    \centering
    \begin{tikzpicture}[
    >=Latex,
    font=\rmfamily\fontsize{8pt}{9pt}\selectfont, 
    line cap=round,
    line join=round,
    node distance=5mm and 2.5mm, 
    box/.style={
      draw=ink,
      rounded corners=1pt,
      line width=0.4pt,
      fill=gfill,
      align=center,
      inner sep=3pt,
      text width=#1
    },
    emphbox/.style={
      draw=edgegray,
      rounded corners=1pt,
      line width=0.6pt,
      fill=gfill2,
      align=center,
      inner sep=3pt,
      text width=#1
    },
    arr/.style={->, draw=linegray, line width=0.6pt}
]

\def\colW{4.05cm} 

\node[box=\colW] (A) {%
\textbf{1) Index-$\nu$ DAE System}\\
$\dot x_d=f_d(x)+g_d(x)u$, $\phi(x)=0$\\
Index $\nu \geq 2$ (Remark~\ref{remark:index_reduction})
};

\node[box=\colW, right=of A] (B) {%
\textbf{2) Successive Differentiation}\\
$\frac{d^k}{dt^k}\phi(x)=0$ for $k=1, \dots, \nu$\\
$\nu$ compatibility conditions
};

\node[box=\colW, below=of A] (C) {%
\textbf{3) Jacobians}\\
Define $J_a^{(\nu)}, J_d^{(\nu)}$ \eqref{eq:top_level_jacobians}\\
Determine control influence $\eta^{(\nu)}$
};

\node[box=\colW, below=of B] (D) {%
\textbf{4) Higher-Index Projection}\\
$P^{(k)}(x) = I - J_a^{(k)}(J_a^{(k)})^\dagger$\\
\eqref{eq:top_level_projection} for $k \in \{1, \dots, \nu\}$
};

\node[box=\colW, below=of C] (E) {%
\textbf{5) Index-$\nu$ Projected Fields}\\
$\hat f(x), \hat g(x)$ using $J_a^{(\nu)}, J_d^{(\nu)}$\\
Lemma~\ref{lemma:projected_fields_indexnu}
};

\node[emphbox=\colW, below=of D] (F) {%
\textbf{6) Manifold \& Safety}\\
(i) \eqref{eq:invariance_projection_indexnu} (Invariance)\\
(ii) \eqref{eq:invariance_cbf_indexnu} (Safety)
};

\node[emphbox=8.4cm, below=6mm of $(E.south)!0.5!(F.south)$, anchor=north] (G) {%
\textbf{Conclusion (Theorem \ref{thm:positive_invariance_indexnu})}: If conditions (i)--(ii) hold for all $k \leq \nu$, then the set $\mathcal{D}\cap\mathcal{M}$ is positively invariant.
};

\draw[arr] (A) -- (B);
\draw[arr] (B.south) -- ++(0,-0.15) -| (C.north);
\draw[arr] (C) -- (D);
\draw[arr] (D.south) -- ++(0,-0.15) -| (E.north);
\draw[arr] (E) -- (F);

\draw[arr] (F.south) -- (F.south |- G.north);

\end{tikzpicture}
        \caption{Roadmap for deriving DAE-aware CBF conditions for higher-index systems ($\nu \geq 2$). The algebraic constraint $\phi(\x)=0$ is differentiated successively up to order $\nu$, eliminating time derivatives via~\eqref{eq:dyn_diff} at each step, until the control input appears explicitly in the extended constraint system. Full row rank of $J_{ext}^{(\nu)}(\x)$ then yields a projected control-affine representation on $\mathcal{M}$, from which the DAE-aware CBF conditions follow analogously to the index-1 case.}
    \label{fig:flowchartHID}
\end{figure}

To characterize the control input that ensures manifold invariance for index-$\nu$ systems, we focus on the $\nu$-th derivative of the algebraic constraint. Define
\begin{equation}
\label{eq:top_level_jacobians}
J_d^{(\nu)}(\x) := \nabla_{\x_d} \left(\frac{d^{\nu}\phi}{dt^{\nu}}\right), \quad J_a^{(\nu)}(\x) := \nabla_{\x_a} \left(\frac{d^{\nu}\phi}{dt^{\nu}}\right).
\end{equation}
We define the projection operator for higher-index systems as
\begin{equation}
\label{eq:top_level_projection}
P^{(\nu)}(\x) := I_{n_m} - J_a^{(\nu)}(\x) \left(J_a^{(\nu)}(\x)\right)^{\dagger},
\end{equation}
which projects onto $\ker\big((J_a^{(\nu)}(\x))^\top\big)$. The compatibility condition $P^{(\nu)}(\x) \big(J_d^{(\nu)}(\x) f_d(\x) + \eta^{(\nu)}(\x) g_d(\x)\ci\big) = 0$ constrains the control input $\ci$ to ensure manifold invariance. The following proposition characterizes the control input that ensures manifold invariance.

\begin{proposition}[Manifold invariance for index-$\nu$ DAEs]
\label{prop:manifold_invariance_indexnu}
Consider an index-$\nu$ DAE system \eqref{eq:dyn} with $\nu \geq 2$ satisfying Assumption~\ref{assumption:regularity}. Let $\x(0) \in \mathcal{M}$ with $\phi^{(k)}(\x(0)) = 0$ for $k = 1, \ldots, \nu$. Then, any trajectory starting from $\x(0)$ remains on $\mathcal{M}$ for all $t \geq 0$ if and only if for all $k\in\{1, \ldots, \nu\}$
\begin{equation}
\label{eq:extended_compatibility}
\frac{d^k}{dt^k} \phi(\x(t)) = 0,
\end{equation}
or equivalently, in projected form, if there exists $\ci$ such that for all $k\in\{1, \ldots, \nu\}$, 
\begin{equation}
\label{eq:extended_compatibility_projected}
P^{(k)}(\x) \big(J_d^{(k)}(\x) f_d(\x) + \eta^{(k)}(\x) g_d(\x) \ci\big) = 0
\end{equation}
\end{proposition}
\begin{proof}
For the state to remain on the manifold, the algebraic constraint $\phi(\x(t)) = 0$ must hold for all $t \geq 0$. This is equivalent to requiring the sequence of derivatives $\frac{d^k \phi}{dt^k}(\x(t)) = 0$ to hold for $k = 1, \ldots, \nu$. 

Differentiating $\phi^{(\nu)}(\x)$ gives
\begin{align*}
\frac{d^\nu}{dt^\nu}\phi(\x(t)) &= J_d^{(\nu)}(\x) \dot{\x}_d + J_a^{(\nu)}(\x) \dot{\x}_a.
\end{align*}
Substituting $\dot{\x}_d = f_d(\x) + g_d(\x)\ci$, we obtain
\begin{equation}
\label{eq:extended_compat_detailed}
J_d^{(\nu)}(\x) \big(f_d(\x) + g_d(\x) \ci\big) + J_a^{(\nu)}(\x) \dot{\x}_a = 0.
\end{equation}
Substituting $\dot{\x}_a^\star(\x,\ci)$ into \eqref{eq:extended_compat_detailed} yields \eqref{eq:extended_compatibility_projected}.
As a result of Assumption~\ref{assumption:regularity}, the extended Jacobian $\begin{bmatrix} J_a^{(\nu)}(\x) & \eta^{(\nu)}(\x) g_d(\x) \end{bmatrix}$ has full row rank. For $k=\nu$, if there exists a solution $\ci$ satisfying~\eqref{eq:extended_compatibility_projected}, we have $\frac{d^\nu}{dt^\nu}\phi(\x(t))=0$. For $k\leq \nu-1$, we have $P^{(k)}(\x) \eta^{(k)}(\x) g_d(\x) \ci = 0$ and the regularity of the system under Assumption~\ref{assumption:regularity} ensures $P^{(k)}(\x) \eta^{(k)}(\x) f_d(\x) = 0$ holds on $\mathcal{M}$.
\end{proof}

The solution for $\dot{\x}_a$ obtained above can be expressed using the Moore--Penrose pseudoinverse of the extended Jacobian. This representation provides a natural construction for expressing the constrained dynamics on $\mathcal{M}$ in control-affine form for index-$\nu$ systems. Analogous to the index-1 case, in the following lemma, we define projected vector fields for index-$\nu$ DAEs.

\begin{lemma}[Projected vector fields for index-$\nu$ DAEs]
\label{lemma:projected_fields_indexnu}
Consider an index-$\nu$ DAE system with $\nu \geq 2$ satisfying Assumption~\ref{assumption:regularity}. Define the projected vector fields
\begin{align}
\label{eq:projected_fields_indexnu}
    \hat{f}(\x) &:=
        \begin{bmatrix}
            f_{d}(\x) \\
            -\left(J_{a}^{(\nu)}(\x)\right)^{\dagger} J_{d}^{(\nu)}(\x) f_{d}(\x)
        \end{bmatrix},
        \\
    \hat{g}(\x) &:=
        \begin{bmatrix}
            g_{d}(\x) \\
            -\left(J_{a}^{(\nu)}(\x)\right)^{\dagger} J_{d}^{(\nu)}(\x) g_{d}(\x)
        \end{bmatrix}.
\end{align}
Then the constrained dynamics on $\mathcal{M}$ admit the control-affine form $\dot{\x} = \hat{f}(\x) + \hat{g}(\x) \ci$, and both $\hat{f}(\x)$ and $\hat{g}(\x)$ are tangent to $\mathcal{M}$ in the sense that they satisfy the projected compatibility condition \eqref{eq:extended_compatibility_projected}.
\end{lemma}

\begin{proof}
From Proposition~\ref{prop:manifold_invariance_indexnu}, manifold invariance requires the compatibility condition~\eqref{eq:extended_compat_detailed}. Substituting the projected algebraic derivative $\dot{\x}_a^{\star}$ into~\eqref{eq:extended_compat_detailed}, we obtain $P^{(k)}(\x) \big(J_d^{(k)}(\x) f_d(\x) + \eta^{(k)}(\x) g_d(\x)\ci\big) = 0$ for all $k\in\{1,\ldots,\nu\}$, which constrain $\ci$ to ensure manifold invariance. Thus any trajectory evolving according to $\dot{\x} = \hat{f}(\x) + \hat{g}(\x)\ci$ with $\ci$ satisfying these constraints remains on $\mathcal{M}$.
\end{proof}

The above lemma establishes that for index-$\nu$ systems, the projected vector fields $\hat{f}(\x)$ and $\hat{g}(\x)$ are constructed using the top-level Jacobians $J_a^{(\nu)}(\x)$ and $J_d^{(\nu)}(\x)$, ensuring that the constrained dynamics remain on the manifold $\mathcal{M}$. This construction motivates the use of high-order control barrier functions, as the control input appears only after multiple differentiations.

\begin{theorem}[Safety of index-$\nu$ DAEs]
\label{thm:positive_invariance_indexnu}
Consider an index-$\nu$ DAE system \eqref{eq:dyn} with $\nu \geq 2$ satisfying Assumption~\ref{assumption:regularity}. Let $b(\x)$ define a safe set $\mathcal{D} = \{\x : b(\x) \geq 0\}$. The set $\mathcal{D} \cap \mathcal{M}$ is positively invariant if there exists a class-$\mathcal{K}_{\infty}$ function $\alpha_0$ and a control input $\ci$ such that for all $\x \in \mathcal{D} \cap \mathcal{M}$ and $k\in\{1,\ldots,\nu\}$,
\begin{subequations}
\label{eq:invariance_conditions_indexnu}
\begin{align}
P^{(k)}(\x) \big(J_{d}^{(k)}(\x) f_{d}(\x) + \eta^{(k)}(\x) g_{d}(\x) \ci\big) &= 0, \label{eq:invariance_projection_indexnu} \\
\mathcal{L}_{\hat{f}} b(\x) + \alpha_0(b(\x)) + \mathcal{L}_{\hat{g}} b(\x) \ci &\geq 0. \label{eq:invariance_cbf_indexnu}
\end{align}
\end{subequations}
\end{theorem}

\begin{proof}
The proof follows the same structure as Theorem~\ref{thm:positive_invariance}. Condition \eqref{eq:invariance_projection_indexnu} ensures manifold invariance via Proposition~\ref{prop:manifold_invariance_indexnu}, so trajectories evolve with dynamics $\dot{\x} = \hat{f}(\x) + \hat{g}(\x)\ci$ on $\mathcal{M}$ by Lemma~\ref{lemma:projected_fields_indexnu}. Condition \eqref{eq:invariance_cbf_indexnu} then ensures $b(\x(t)) \geq 0$ via the standard CBF comparison lemma argument on the projected dynamics.
\end{proof}
This theorem reveals a fundamental structure for index-$\nu$ systems: while the algebraic constraint $\phi$ requires high-order treatment due to the index-$\nu$ structure (Proposition~\ref{prop:manifold_invariance_indexnu} requires enforcing the $\nu$-th order compatibility condition), the barrier function $b$ is enforced via a standard CBF condition \eqref{eq:invariance_cbf_indexnu} once the trajectory is constrained to $\mathcal{M}$. The projected vector fields $\hat{f}(\x)$ and $\hat{g}(\x)$ incorporate all higher-order hidden algebraic constraints, allowing safety to be enforced using CBF methods applied on the manifold.

\begin{definition}[DAE-aware CBF]\label{def:dae_aware_cbf}
A continuous and differentiable function $b$ on $\mathcal{M}$ is called a \emph{DAE-aware CBF} if there exists a function $\alpha_{0}$ such that, for every $\x \in \mathcal{D} \cap \mathcal{M}$, there exists a control input $\ci \in \mathcal{U}$ satisfying \eqref{eq:invariance_conditions_indexnu}.
\end{definition}

Theorems~\ref{thm:positive_invariance} and~\ref{thm:positive_invariance_indexnu} establish sufficient conditions for positive invariance of a safe set $\mathcal{C}$ under the DAE dynamics~\eqref{eq:dyn}, assuming that $b$ is a valid DAE-aware CBF. The following section develops computational methods to verify this assumption for candidate barrier functions represented as polynomials or neural networks.  

\section{DAE-aware CBF Verification}
\label{section:verification}

In this section, we focus on verifying if the given function $b : \mathbb{R}^{n_x} \to \mathbb{R}$ is a valid DAE-aware CBF. 
The verification consists of geometric 
correctness and feasibility checking. Geometric correctness (referred to as correctness) ensures that the safe set $\mathcal{C}$ contains the candidate barrier set $\mathcal{D}$ on the manifold, i.e., $\mathcal{D} \cap \mathcal{M} \subseteq \mathcal{C} \cap \mathcal{M}$. Feasibility ensures that for every $\x \in \mathcal{D} \cap \mathcal{M}$ there exists a control input $\ci$ that preserves both manifold compatibility and positive invariance.

\begin{problem}
\label{problem:verification}
Given a candidate DAE-aware CBF $b : \mathbb{R}^{n_x} \to \mathbb{R}$, verify the following
\begin{enumerate}
    \item (Correctness) $\mathcal{D}\cap\mathcal{M} \subseteq \mathcal{C}\cap\mathcal{M}$.
    \item (Feasibility) For every $\x \in \mathcal{D} \cap \mathcal{M}$, there exists a control input $\ci$ such that the compatibility condition holds and the set $\mathcal{D} \cap \mathcal{M}$ is positively invariant.
\end{enumerate}
\end{problem}

The proposed verification framework is shown in Fig. \ref{fig:roadmap_verification}.

\subsection{Preliminaries}
The results in this section concern Positivstellensatz and Farkas' Lemma, and are standard in real algebraic geometry; see~\cite{parrilo2003semidefinite,matousek2006understanding} for more details. These results will be used to derive the verification conditions for safety CBFs.

We first introduce some required background from real algebraic geometry. A polynomial, denoted as $p(\x)$, is SOS if and only if it can be written as 
\begin{equation}
    p(\x) = \sum_{i=1}^{n} (p_i(\x))^2 \nonumber
\end{equation}
for some polynomials $p_i(\x)$. The set of all SOS polynomials is denoted as $\mbox{SOS}$.

For certifying non-negative polynomials, the cone generated from a set of polynomials $\Phi_{1},\ldots,\Phi_{k_{\Sigma}}$ is given by 
\begin{equation*}
\Sigma[\Phi_{1},\ldots,\Phi_{k_{\Sigma}}] = \left\{\sum_{K \subseteq \{1,\ldots,k_{\Sigma}\}}{\gamma_{K}(\x)\prod_{i \in K}{\Phi_{i}(\x)}} : \gamma_{K}(\x) \in \mbox{ SOS } \;\forall \; K \subseteq \{1,\ldots,k_{\Sigma}\}\right\}.
\end{equation*}
The monoid $\mathcal{N}$ generated from a set of polynomials $\chi_{1},\ldots,\chi_{k_{\mathcal{N}}}$ with non-negative integer exponents $r_{1}\ldots r_{k_{\mathcal{N}}}$ is defined as
\begin{equation*}
\mathcal{N}[\chi_{1},\ldots,\chi_{k_{\mathcal{N}}}] = \left\{{\prod_{i =1}^{k_{\mathcal{N}}}{\chi^{r_{i}}(\x)}} : r_{1}\ldots r_{k_{\mathcal{N}}}\in\mathbb{Z}_{+}\right\}.
\end{equation*}
The ideal generated from polynomials $\Gamma_{1},\ldots,\Gamma_{k_{\mathbb{I}}}$ is given by 
\begin{equation*}
\mathbb{I}[\Gamma_{1},\ldots,\Gamma_{k_{\mathbb{I}}}] = \left\{\sum_{i=1}^{k_{\mathbb{I}}}{p_{i}(\x)\Gamma_{i}(\x)} : p_{1},\ldots,p_{k_{\mathbb{I}}} \mbox{ are polynomials}\right\}.
\end{equation*}

The following theorem shows that the emptiness of a set defined by polynomial inequalities, non-vanishing conditions, and equalities is equivalent to the existence of polynomials belonging to the respective cone, monoid, and ideal, whose sum equals zero. %

\begin{theorem}[Positivstellensatz \cite{parrilo2003semidefinite}]
\label{th:Positivstellensatz}
Let $\left(\Phi_{i}\right)_{j=1, \ldots, k_{\Sigma}}$, $\left(\chi_{j}\right)_{k=1, \ldots, k_{\mathcal{N}}}$, $\left(\Gamma_{\ell}\right)_{\ell=1, \ldots, k_{\mathbb{I}}}$ be finite families of polynomials. Then, the following properties are equivalent.
\begin{enumerate}
    \item The set
    \begin{equation}
    \label{eq:psatz_set}
    \left\{
    \begin{array}{l|ll}
    \x \in \mathbb{R}^{n_{x}} &
    \begin{array}{ll}
    \Phi_{i}(\x) \geq 0, & i=1, \ldots, k_{\Sigma} \\
    \chi_{j}(\x) \neq 0, & j=1, \ldots, k_{\mathcal{N}} \\
    \Gamma_{l}(\x)=0, & l=1, \ldots, k_{\mathbb{I}}
    \end{array}
    \end{array}
    \right\}
    \end{equation}
    is empty.
    \item There exist $\Phi \in \Sigma, \chi \in \mathcal{N}, \Gamma \in \mathbb{I}$ such that $\Phi+\chi^{2}+\Gamma=0$.
\end{enumerate}
\end{theorem}
Positivstellensatz is a central theorem in real algebraic geometry; it provides a certificate of infeasibility for systems defined by polynomial inequalities, equalities, and non-vanishing constraints. When the polynomials are linear, a classical result that plays an analogous role is Farkas' Lemma. The following is one of its variants; several versions of this lemma exist. It is used, in subsequent sections, to derive the synthesis conditions for safety CBFs.

\begin{figure}[t]
    \centering
    \begin{tikzpicture}[
    >=Latex,
    font=\rmfamily\fontsize{8pt}{9pt}\selectfont, 
    line cap=round,
    line join=round,
    node distance=4mm and 3mm, 
    box/.style={
      draw=ink,
      rounded corners=1pt,
      line width=0.4pt,
      fill=gfill,
      align=center,
      inner sep=3pt,
      text width=#1
    },
    emphbox/.style={
      draw=edgegray,
      rounded corners=1pt,
      line width=0.6pt,
      fill=gfill2,
      align=center,
      inner sep=3pt,
      text width=#1
    },
    arr/.style={->, draw=linegray, line width=0.6pt}
]

\def\colW{3.8cm} 

\node[emphbox=8cm] (Start) {%
\textbf{Problem \ref{problem:verification}: Verification of Candidate CBF $b(x)$}\\
Is the function Correct? (Goal alignment) \quad \& \quad Is it Feasible? (Control existence)
};

\node[box=\colW, below=5mm of Start.south west, anchor=north west] (C1) {%
\textbf{1) Correctness Check}\\
$\mathcal{D} \cap \mathcal{M} \subseteq \mathcal{C} \cap \mathcal{M}$\\
Ensures $b(x) \geq 0 \Rightarrow h(x) \geq 0$
};

\node[box=\colW, below=of C1] (C2) {%
\textbf{Certificate Search}\\
\textit{Polynomial:} SOS Program \eqref{eq:crt_sos_program}\\
\textit{General:} NLP Min-Search \eqref{eq:correctness_nlp}
};

\node[box=\colW, below=5mm of Start.south east, anchor=north east] (F1) {%
\textbf{2) Feasibility Check}\\
Does an admissible $u$ exist for\\
compatibility + CBF + input bounds?
};

\node[box=\colW, below=of F1] (F2) {%
\textbf{Farkas Duality}\\
Form $A_{\mathrm{dae}}$ and $r_{\mathrm{dae}}$ \eqref{eq:dae_affine_matrix}-\eqref{eq:dae_affine_vector}\\
Convert to dual multipliers $\lambda$
};

\node[box=1.8cm, below left=4mm and -19mm of F2] (F3a) {%
\textbf{Interior}\\
Check \eqref{eq:dae_dual_int}
};
\node[box=1.8cm, below right=4mm and -19mm of F2] (F3b) {%
\textbf{Boundary}\\
Check \eqref{eq:dae_dual_bnd}
};

\node[box=8cm, below=2.5cm of C2.west, anchor=north west] (M1) {%
\textbf{Implementation Framework}\\
\textit{Polynomials (Assumption \ref{assumption:polynomial_dae}):} Positivstellensatz Certificates (Prop. \ref{prop:polynomial_dae_aware})\\
\textit{Non-Polynomials:} Global NLP for Counterexamples $x_{ce}$
};

\node[emphbox=8cm, below=4mm of M1] (End) {%
\textbf{Theorem \ref{thm:dae_aware_cbf_verification} Result}\\
If Correctness, Interior Feasibility, and Boundary Tangency hold\\
$\Rightarrow$ $b(x)$ is a valid \textbf{DAE-aware CBF}.
};

\draw[arr] (Start.south -| C1) -- (C1);
\draw[arr] (Start.south -| F1) -- (F1);
\draw[arr] (C1) -- (C2);
\draw[arr] (F1) -- (F2);
\draw[arr] (F2) -- (F3a);
\draw[arr] (F2) -- (F3b);

\draw[arr] (C2.south) -| (M1.160);
\draw[arr] (F3a.south) -- ++(0,-0.15) -| (M1.20);
\draw[arr] (F3b.south) -- ++(0,-0.15) -| (M1.20);

\draw[arr] (M1) -- (End);

\end{tikzpicture}

\caption{Roadmap for the verification of a candidate DAE-aware CBF $b(\x)$ (Problem~\ref{problem:verification}). The framework consists of \textit{Left:} the correctness branch verifies that $\mathcal{D} \cap \mathcal{M} \subseteq \mathcal{C} \cap \mathcal{M}$ via an SOS program~\eqref{eq:crt_sos_program} for polynomial systems, or a nonlinear program (NLP) counterexample search~\eqref{eq:correctness_nlp} in the general case. \textit{Right:} the feasibility branch applies Farkas duality to the affine input constraint system $(A_{\mathrm{dae}}, \boldsymbol{r}_{\mathrm{dae}})$, decomposing the feasibility certificate into an interior condition~\eqref{eq:dae_dual_int} and a boundary tangency condition~\eqref{eq:dae_dual_bnd}. Both branches are unified in Theorem~\ref{thm:dae_aware_cbf_verification}: if correctness, interior feasibility, and boundary tangency all hold, then $b(\x)$ is certified as a valid DAE-aware CBF.}
\label{fig:roadmap_verification}
\end{figure}
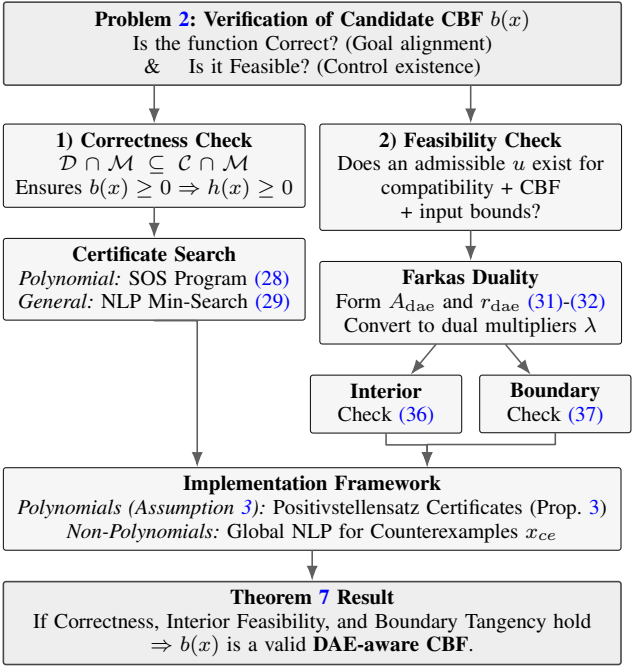

\begin{lemma}[Farkas' Lemma \cite{matousek2006understanding}]
\label{Lemma:Farkas}
Let $A \in \mathbb{R}^{m \times n}$ be a real matrix with $m$ rows and $n$ columns, and let ${r} \in \mathbb{R}^{n}$ be a vector. One of the following conditions holds.
\begin{itemize}
    \item The system of inequalities $A \x \leq {r}$ has a solution.
    \item There exists ${y}$ such that ${y} \geq {0}$, ${y}^{T} A={0}^{T}$ and ${y}^{T} {r} < 0$, where ${0}$ denotes a zero vector.
\end{itemize}
\end{lemma}

Farkas' Lemma states that for a given system of linear inequalities, either $(i)$ the system has a feasible solution, or $(ii)$ there exists a non-negative vector ${y}$ that certifies the infeasibility of the system. This lemma is a key result in linear programming for analyzing the feasibility of linear systems.

\subsection{Correctness Verification}
We first present the correctness condition. Correctness requires $\mathcal{D} \cap \mathcal{M} \subseteq \mathcal{C} \cap \mathcal{M}$ so that any state on the manifold considered safe by the candidate barrier set is also safe under the specification $h(\x)$. A sufficient condition that we enforce on the admissible manifold is
\begin{equation}
\label{eq:correctness_pointwise}
    h(\x) \; \geq \; 0, \qquad \forall\, \x \in \mathcal{D} \cap \mathcal{M}.
\end{equation}

Take any $\x \in \mathcal{D} \cap \mathcal{M}$. Then $b(\x) \ge 0$ by definition and \eqref{eq:correctness_pointwise} implies $h(\x) \ge 0$, so $\x \in \mathcal{C}\cap\mathcal{M}$. The inequality \eqref{eq:correctness_pointwise} can be checked or certified in two complementary ways, i.e., a Positivstellensatz-based approach for polynomial systems and nonlinear programs.

We denote a state violating \eqref{eq:correctness_pointwise} as a correctness counterexample $\x^{(c)}_{ce}$ and the set of correctness counterexamples as $\mathbf{X}^{(c)}_{ce}$. We introduce a slack variable $s \neq 0$ and define the set of violating states over variables $(\x,s)$ as follows:
$$
\mathcal{S}_{\mathrm{viol}} := \big\{(\x,s) : \phi(\x)={0}, b(\x) \ge 0, -h(\x) - s^{2} \ge 0, s \neq 0\big\}
$$
The following lemma characterizes correctness using Positivstellensatz.

\begin{lemma}[Correctness via Positivstellensatz]
\label{lem:crt_psatz}
Assume $h$, $b$, and $\phi$ are polynomials and the set $\mathcal{D} \cap \mathcal{M}$ is compact. The given CBF is geometrically correct (i.e., $\mathcal{D} \cap \mathcal{M} \subseteq \mathcal{C} \cap \mathcal{M}$) if there exist SOS polynomials $\sigma_0, \sigma_1, \sigma_2 \in \Sigma[\x, s]$, polynomials $\rho_\ell \in \mathbb{R}[\x, s]$ for $\ell = 1, \ldots, m$, and an integer $k \geq 1$ such that
\begin{equation}
\label{eq:crt_psatz_identity}
    \Phi(\x, s) + \chi(\x, s) + \varphi(\x, s) = 0,
\end{equation}
where
\begin{equation*}
    \Phi(\x, s) = \sigma_0 + \sigma_1 \cdot b(\x) + \sigma_2 \cdot (-h(\x) - s^2), \quad \chi(\x, s) = s^{2k}, \quad \varphi(\x, s) = \sum_{\ell=1}^{m} \rho_\ell(\x, s) \, \phi_\ell(\x).
\end{equation*}
\end{lemma}
\begin{proof}
We show that identity~\eqref{eq:crt_psatz_identity} certifies $\mathcal{D} \cap \mathcal{M} \subseteq \mathcal{C} \cap \mathcal{M}$ by proving the semialgebraic set of violations $\mathcal{S}_{\mathrm{viol}}$ empty. 
The violation set is not empty, i.e., $\mathcal{S}_{\mathrm{viol}} \neq \emptyset$, if and only if there exists $\x$ with $\phi(\x) = {0}$, $b(\x) \ge 0$, and $h(\x) < 0$. We note that the slack variable $s$ with constraint $s \neq 0$ encodes the strict inequality $h(\x) < 0$. 

By Theorem~\ref{th:Positivstellensatz}, the set $\mathcal{S}_{\mathrm{viol}}$ is empty if and only if there exist $\Phi \in \Sigma[b, -h-s^2]$, $\chi \in \mathcal{M}[s]$, and $\varphi \in \mathbb{I}[\phi_1, \ldots, \phi_m]$ such that $\Phi + \chi^2 + \varphi = 0$. Writing $\chi = s^k$ for some integer $k \geq 1$, we have $\chi^2 = s^{2k}$, which yields identity~\eqref{eq:crt_psatz_identity}.

Hence the existence of such SOS polynomials $\sigma_0, \sigma_1, \sigma_2$, polynomial multipliers $\rho_\ell$, and integer $k \geq 1$ certifies that no $\x \in \mathcal{D} \cap \mathcal{M}$ satisfies $h(\x) < 0$, i.e., $\mathcal{D} \cap \mathcal{M} \subseteq \mathcal{C} \cap \mathcal{M}$.
\end{proof}

The correctness verification can be formulated as the following SOS program.

\begin{align}
    \text{Find } &\sigma_0(\x,s), \sigma_1(\x,s), \sigma_2(\x,s), \rho_1(\x,s), \ldots, \rho_m(\x,s), k \nonumber\\
    \text{s.t. } & \sigma_0, \sigma_1, \sigma_2 \in \text{SOS}, \quad k \geq 1, \nonumber\\
    & -1 + \sigma_0 + \sigma_1 \cdot b(\x) + \sigma_2 \cdot (-h(\x) - s^2) + s^{2k} + \sum_{\ell=1}^{m} \rho_\ell \, \phi_\ell(\x) \in \text{SOS}.
\label{eq:crt_sos_program}
\end{align}
If the program~\eqref{eq:crt_sos_program} is feasible, then $\mathcal{D} \cap \mathcal{M} \subseteq \mathcal{C} \cap \mathcal{M}$.

The above Positivstellensatz 
condition \eqref{eq:crt_psatz_identity} is a sufficient and necessary verification. The SOS program~\eqref{eq:crt_sos_program} is a sufficient certificate for correctness.

\parstart{Revisiting Example~\ref{exmpl:DAE_unaware_example_1}}
For the wind turbine~\eqref{eq:wind_turbine}, program~\eqref{eq:crt_sos_program} is instantiated with $h(\x) = x_1 - x_{\max}$, the degree-2 candidate $b(\x)$ from Section~\ref{sec:experiments}, and the algebraic constraint $\phi(\x)$ from~\eqref{eq:wind_turbine} ($m=1$). With degree-2 SOS multipliers, the program is feasible, certifying $\mathcal{D} \cap \mathcal{M} \subseteq \mathcal{C} \cap \mathcal{M}$; see Tab.~\ref{tab:windturbine_size_time}.

If $h$ or $b$ is non-polynomial, we can verify \eqref{eq:correctness_pointwise} by solving
\begin{equation}
    \label{eq:correctness_nlp}
\begin{array}{ll}
\min\limits_{\x} & h(\x) \\
\text{s.t.} & \phi(\x) = {0}, \\
& b(\x) \ge 0.
\end{array}
\end{equation}
If the optimal value is nonnegative, then \eqref{eq:correctness_pointwise} holds. Otherwise, the optimizer provides a counterexample $\x^{(c)}_{ce}$ with $\phi(\x)={0}$, $b(\x) \ge 0$, and $h(\x) < 0$. 

\begin{remark}
\label{remark:nlp_verification}
Nonlinear programming verifications require global optimality to avoid false negatives. While global certificates can be obtained for polynomial instances via SMT or mixed-integer relaxations, such approaches may not scale efficiently. A negative optimal value yields a concrete counterexample demonstrating the violation.
\end{remark}

\subsection{Feasibility Verification}
\label{subsec:feasibility}
We next present the feasibility conditions for the DAE-aware CBFs. The analysis is divided into two parts, addressing the interior and boundary conditions separately.

Having established the projected vector field framework for index-1 DAEs and the compatibility conditions for higher-index DAEs, we now present unified feasibility conditions that apply to both cases. The key challenge is to ensure that the CBF condition can be satisfied while respecting both the algebraic constraints and any bounds on the control input.

In what follows, we denote by $n_{u}^{c} \in \mathbb{N}$ the number of linear constraints imposed on the control input. For simplicity, we unify the notations for index-1 and index-$\nu$ cases. Let $\bar{P}(\x)$, $\bar{J}_d(\x)$, and $\bar{\eta}(\x)$ denote $P(\x)$, $J_d(\x)$, and $\eta^{(1)}(\x)$ (for index-1) or $P^{(k)}(\x)$, $J_d^{(k)}(\x)$, and $\eta^{(k)}(\x)$ for all $k\in\{1, \ldots, \nu\}$ (for index-$\nu$). For a given $\x \in \mathcal{D} \cap \mathcal{M}$, the existence of a control input $\ci$ satisfying the compatibility, CBF, and input constraints
\begin{equation}
\begin{aligned}
    \bar{P}(\x) \big(\bar{J}_{d}(\x) f_{d}(\x) + \bar{\eta}(\x) g_{d}(\x) \ci\big) &= {0}, \\
    \mathcal{L}_{\hat{f}} b(\x) + \alpha\big(b(\x)\big)
        + \mathcal{L}_{\hat{g}} b(\x) \, \ci &\geq 0, \\
    A_{u} \ci &\leq {r}_{u},
\end{aligned}
    \label{eq:dae_sufficient_conditions}
\end{equation}
is equivalent to the feasibility of the affine system $A_{\mathrm{dae}}(\x) \ci \leq {r}_{\mathrm{dae}}(\x)$, where
\begin{align}
    A_{\mathrm{dae}}(\x) &:=
        \begin{bmatrix}
            -\bar{P}(\x) \bar{\eta}(\x) g_{d}(\x) \\
            \bar{P}(\x) \bar{\eta}(\x) g_{d}(\x) \\
            -\mathcal{L}_{\hat{g}} b(\x) \\
            A_{u}
        \end{bmatrix},
        \label{eq:dae_affine_matrix} \\
    {r}_{\mathrm{dae}}(\x) &:=
        \begin{bmatrix}
            \bar{P}(\x) \bar{J}_{d}(\x) f_{d}(\x) \\
            -\bar{P}(\x) \bar{J}_{d}(\x) f_{d}(\x) \\
            \mathcal{L}_{\hat{f}} b(\x) + \alpha\big(b(\x)\big) \\
            {r}_{u}
        \end{bmatrix}.
    \label{eq:dae_affine_vector}
\end{align}
This equivalence follows from encoding the compatibility equality as two opposing inequalities $\pm \bar{P}(\x) \bar{J}_{d}(\x)(f_{d}(\x) + g_{d}(\x) \ci) \geq {0}$. The first two block rows enforce the compatibility condition, the third row enforces the CBF condition, and the last block row enforces the input constraint $\ci \in \mathcal{U}$.

The following proposition provides a dual characterization of feasibility.
\begin{proposition}[Dual feasibility condition]
\label{prop:dae_feasibility}
The system~\eqref{eq:dae_sufficient_conditions} admits a solution $\ci$ if and only if there does not exist a multiplier ${\lambda} \in \mathbb{R}_{+}^{2 n_{m}+1+n_{u}^{c}}$ such that
\begin{equation}
    {\lambda}^{\top} A_{\mathrm{dae}}(\x) = {0}^{\top},
    \qquad
    {\lambda}^{\top} {r}_{\mathrm{dae}}(\x) < 0.
    \label{eq:dae_dual_vector_form}
\end{equation}
Equivalently, no nonnegative multipliers ${\lambda}^{-}, {\lambda}^{+} \in \mathbb{R}_{+}^{n_{m}}$, $\lambda^{b} \in \mathbb{R}_{+}$, and ${\lambda}^{u} \in \mathbb{R}_{+}^{n_{u}^{c}}$ satisfy
\begin{align}
    &\big(\bar{P}(\x) \bar{\eta}(\x) g_{d}(\x)\big)^{T}
        \big({\lambda}^{+} - {\lambda}^{-}\big)
        - \mathcal{L}_{\hat{g}} b(\x) \, \lambda^{b} + A_{u}^{T} {\lambda}^{u} = {0},
        \label{eq:dae_dual_balance} \\
    &\big({\lambda}^{-} - {\lambda}^{+}\big)^{T}
        \bar{P}(\x) \bar{J}_{d}(\x) f_{d}(\x)
        + \big({\lambda}^{u}\big)^{T} {r}_{u} + \lambda^{b} \Big(\mathcal{L}_{\hat{f}} b(\x) + \alpha\big(b(\x)\big)\Big) < 0.
        \label{eq:dae_dual_certificate}
\end{align}
\end{proposition}
\begin{proof}
By Lemma~\ref{Lemma:Farkas}, the system $A_{\mathrm{dae}}(\x) \ci \leq {r}_{\mathrm{dae}}(\x)$ has no solution if and only if there exists ${\lambda} \geq {0}$ such that ${\lambda}^{\top} A_{\mathrm{dae}}(\x) = {0}^{\top}$ and ${\lambda}^{\top} {r}_{\mathrm{dae}}(\x) < 0$. Partitioning ${\lambda} = \big[\big({\lambda}^{-}\big)^{\top},\big({\lambda}^{+}\big)^{\top}, \lambda^{b}, \big({\lambda}^{u}\big)^{\top}\big]^{\top}$ according to the block structure of $A_{\mathrm{dae}}(\x)$ in~\eqref{eq:dae_affine_matrix}, the condition ${\lambda}^{\top} A_{\mathrm{dae}}(\x) = {0}^{\top}$ expands to~\eqref{eq:dae_dual_balance}, and ${\lambda}^{\top} {r}_{\mathrm{dae}}(\x) < 0$ expands to~\eqref{eq:dae_dual_certificate}.
\end{proof}

Proposition~\ref{prop:dae_feasibility} provides a dual characterization for the primal feasibility problem~\eqref{eq:dae_sufficient_conditions}. However, this formulation requires a class-$\mathcal{K}$ function $\alpha$ to be specified, and the resulting CBF condition is sufficient but not necessary for safety since $\alpha$ regulates how aggressively the control input must steer the system away from the boundary. In what follows, we derive necessary and sufficient conditions for positive invariance based on Nagumo's theorem, which requires only compatibility at interior points and tangency at boundary points. These exact conditions provide maximal flexibility in control design by characterizing the full set of admissible inputs.

Interior feasibility relies only on compatibility and input constraints. We define the interior matrices
\begin{align*}
    A_{\mathrm{int}}(\x) &:=
        \begin{bmatrix}
            -\bar{P}(\x) \bar{\eta}(\x) g_{d}(\x) \\
            \bar{P}(\x) \bar{\eta}(\x) g_{d}(\x) \\
            A_{u}
        \end{bmatrix},\qquad
    {r}_{\mathrm{int}}(\x) &:=
        \begin{bmatrix}
            \bar{P}(\x) \bar{J}_{d}(\x) f_{d}(\x) \\
            -\bar{P}(\x) \bar{J}_{d}(\x) f_{d}(\x) \\
            {r}_{u}
        \end{bmatrix},
\end{align*}
and the feasibility of
\begin{equation*}
    A_{\mathrm{int}}(\x) \, \ci \; \le \; {r}_{\mathrm{int}}(\x)
\end{equation*}
enforces $\bar{P}(\x) \big(\bar{J}_{d}(\x) f_{d}(\x) + \bar{\eta}(\x) g_{d}(\x) \ci\big) = {0}$ and $A_{u} \ci \le {r}_{u}$ for all interior points.

For boundary points, we collect compatibility, tangency, and input constraints as
\begin{align*}
    A_{\mathrm{bnd}}(\x) &:=
        \begin{bmatrix}
            A_{\mathrm{int}}(\x) \\
            -\mathcal{L}_{\hat{g}} b(\x) 
        \end{bmatrix}, &
    {r}_{\mathrm{bnd}}(\x) &:=
        \begin{bmatrix}
            {r}_{\mathrm{int}}(\x) \\
            \mathcal{L}_{\hat{f}} b(\x) 
        \end{bmatrix}.
\end{align*}
Thus, the boundary inequalities $A_{\mathrm{bnd}}(\x) \ci \leq {r}_{\mathrm{bnd}}(\x)$ encode
\begin{equation*}
    \begin{aligned}
        \bar{P}(\x) \big(\bar{J}_{d}(\x) f_{d}(\x) + \bar{\eta}(\x) g_{d}(\x) \ci\big) &= {0}, \\
        \mathcal{L}_{\hat{f}} b(\x) + \mathcal{L}_{\hat{g}} b(\x) \, \ci &\geq 0, \\
        A_{u} \ci &\leq {r}_{u}.
    \end{aligned}
\end{equation*}
By Lemma~\ref{Lemma:Farkas}, feasibility of the interior and boundary stacks is equivalent to infeasibility of the dual systems
\begin{align}
    {\lambda}^{\top} A_{\mathrm{int}}(\x) = {0}^{\top},
    &\qquad {\lambda}^{\top} {r}_{\mathrm{int}}(\x) < 0,
        \label{eq:dae_dual_int} \\
    {\lambda}^{\top} A_{\mathrm{bnd}}(\x) = {0}^{\top},
    &\qquad {\lambda}^{\top} {r}_{\mathrm{bnd}}(\x) < 0,
        \label{eq:dae_dual_bnd}
\end{align}
with nonnegative multipliers ${\lambda} = \big[\big({\lambda}^{-}\big)^{\top},\big({\lambda}^{+}\big)^{\top}, \lambda^{b}, \big({\lambda}^{u}\big)^{\top}\big]^{\top}$ (and $\lambda^{b}$ omitted in the interior). Expanding \eqref{eq:dae_dual_bnd} gives
\begin{align}
\big(\bar{P}(\x) \bar{\eta}(\x) g_{d}(\x)\big)^{T}\big({\lambda}^{+} - {\lambda}^{-}\big)
 - \big(\mathcal{L}_{\hat{g}} b(\x)\big)^{T} \lambda^{b} + A_{u}^{T} {\lambda}^{u} &= {0},
    \label{eq:dae_dual_bnd_balance} \\
\big({\lambda}^{-} - {\lambda}^{+}\big)^{T} \bar{P}(\x) \bar{J}_{d}(\x) f_{d}(\x)
    + \lambda^{b} \, \mathcal{L}_{\hat{f}} b(\x)
    + \big({\lambda}^{u}\big)^{T} {r}_{u} &< 0.
            \label{eq:dae_dual_bnd_certificate}
\end{align}

\parstart{Revisiting Example~\ref{exmpl:DAE_unaware_example_1}}
For the wind turbine~\eqref{eq:wind_turbine}, the matrices $A_{\mathrm{int}}(\x)$ and $A_{\mathrm{bnd}}(\x)$ are constructed from the projected dynamics $(\hat{f}, \hat{g})$ computed in Section~\ref{sec:dae_cbf} and the input bounds $A_u \ci \leq {r}_u$. With degree-2 SOS multipliers, certificates are successfully found showing that the solution set of $\ci$ is not empty. Hence, a safe input $\ci$ exists at every state $\x\in\mathcal{D} \cap \mathcal{M}$; see Tab.~\ref{tab:windturbine_size_time}.

\begin{theorem}[DAE-aware CBF verification]
\label{thm:dae_aware_cbf_verification}
Let $b$ define the barrier set
\( \mathcal{D} := \{\x : b(\x)\ge 0\} \). Then $b$ is a \emph{DAE-aware CBF} if and only if the following conditions hold
\begin{enumerate}
    \item (Correctness) $\mathcal{D} \cap \mathcal{M} \subseteq \mathcal{C} \cap \mathcal{M}$.
    \item (Interior feasibility) For every $\x \in \operatorname{int}(\mathcal{D}) \cap \mathcal{M}$, the dual system \eqref{eq:dae_dual_int} admits no nonnegative solution.
    \item (Boundary tangency) For every $\x \in \partial \mathcal{D} \cap \mathcal{M}$, the dual system \eqref{eq:dae_dual_bnd} admits no nonnegative solution.
\end{enumerate}
\end{theorem}
\begin{proof}
By Lemma~\ref{Lemma:Farkas}, $A_{\mathrm{int}}(\x)\,\ci \le {r}_{\mathrm{int}}(\x)$ is feasible if and only if~\eqref{eq:dae_dual_int} has no nonnegative solution, and $A_{\mathrm{bnd}}(\x)\,\ci \le {r}_{\mathrm{bnd}}(\x)$ is feasible if and only if~\eqref{eq:dae_dual_bnd} has no nonnegative solution. Conditions~2) and~3) thus guarantee the existence of admissible inputs satisfying compatibility and input bounds at interior points, and additionally tangency at boundary points. Positive invariance of $\mathcal{D} \cap \mathcal{M}$ then follows from Theorems~\ref{thm:positive_invariance} and~\ref{thm:positive_invariance_indexnu}.
\end{proof}

The interior and boundary dual systems \eqref{eq:dae_dual_int} and \eqref{eq:dae_dual_bnd} are affine in the multipliers ${\lambda}$; infeasibility can therefore be certified using the framework of Section~\ref{sec:problem}. We next present two approaches to certify infeasibility, one for polynomial systems and one for non-polynomial systems.

\subsubsection{Verification for Polynomial Systems}
Consider the case where a candidate DAE-aware CBF $b : \mathbb{R}^{n_x} \to \mathbb{R}$ is given. In what follows, we make the following assumption on the system dynamics and the sets $\mathcal{M}$ and $\mathcal{C}$. 

\begin{assumption}
\label{assumption:polynomial_dae}
The functions $f_{d}$, $g_{d}$, $J_{d}$, $P$, and $b$ are polynomial functions of the state variables. The constraint manifold $\mathcal{M}$ and the safe set $\mathcal{C}$ are described by polynomial equalities and inequalities.
\end{assumption}

Under this assumption, every entry of the matrices $A_{\mathrm{int}}(\x)$, $A_{\mathrm{bnd}}(\x)$ and vectors ${r}_{\mathrm{int}}(\x)$, ${r}_{\mathrm{bnd}}(\x)$ used in Theorem~\ref{thm:dae_aware_cbf_verification} is polynomial in $\x$. We present SOS certificates based on the dual infeasibility conditions \eqref{eq:dae_dual_int} and \eqref{eq:dae_dual_bnd}.

We now formulate SOS certificates 
based on the dual infeasibility conditions 
\eqref{eq:dae_dual_int} and \eqref{eq:dae_dual_bnd}. To apply Theorem~\ref{th:Positivstellensatz}, we formulate emptiness problems for the interior and boundary dual systems over variables $(\x, {\lambda}, s)$. The \emph{interior cone} $\Sigma_{\mathrm{int}}$ is generated by the inequalities $b(\x)\ge 0$ and componentwise nonnegativity of ${\lambda}_{\mathrm{int}}$, whereas the \emph{boundary cone} $\Sigma_{\mathrm{bnd}}$ is generated by the nonnegativity of ${\lambda}_{\mathrm{bnd}}$. 

Strict inequalities ${\lambda}^{\top} {r}_{\mathrm{int}}(\x) < 0$ and ${\lambda}^{\top} {r}_{\mathrm{bnd}}(\x) < 0$ are encoded via slack variables $s_{\mathrm{int}}^{2}$ and $s_{\mathrm{bnd}}^{2}$ in the respective monoids. The ideals $\small \mathbb{I}_{\mathrm{int}}:= \big\langle {\lambda}^{\top}A_{\mathrm{int}}(\x),\,\phi(\x)\big\rangle$ and $\mathbb{I}_{\mathrm{bnd}}:= \big\langle {\lambda}^{\top}A_{\mathrm{bnd}}(\x),\,\phi(\x),\,b(\x)\big\rangle$ are generated by the equalities ${\lambda}^{\top} A_{\mathrm{int}}(\x)={0}^{\top}$, $\phi(\x)={0}$ (interior), and ${\lambda}^{\top} A_{\mathrm{bnd}}(\x)={0}^{\top}$, $\phi(\x)={0}$, $b(\x)=0$ (boundary).

The emptiness of the interior and boundary dual systems 
can thus be certified through SOS representations 
of the type established by the Positivstellensatz. We formalize these SOS conditions accordingly.
\begin{proposition}[Polynomial certificates for DAE-aware CBFs]
\label{prop:polynomial_dae_aware}
Suppose Assumption~\ref{assumption:polynomial_dae} holds. The interior feasibility condition 2) and boundary tangency condition 3) of Theorem~\ref{thm:dae_aware_cbf_verification} hold if there exist SOS polynomials and polynomial multipliers satisfying
\begin{equation*}
\Phi_{\mathrm{int}}(\x, {\lambda}_{\mathrm{int}}) + s_{\mathrm{int}}^{2} + \varphi_{\mathrm{int}}(\x, {\lambda}_{\mathrm{int}}) = 0,
\end{equation*}
with $\Phi_{\mathrm{int}} \in \Sigma_{\mathrm{int}}$, $s_{\mathrm{int}}^{2}$ in the monoid, and
\begin{equation*}
\Phi_{\mathrm{bnd}}(\x, {\lambda}_{\mathrm{bnd}}) + s_{\mathrm{bnd}}^{2} + \varphi_{\mathrm{bnd}}(\x, {\lambda}_{\mathrm{bnd}}) = 0,
\end{equation*}
with $\Phi_{\mathrm{bnd}} \in \Sigma_{\mathrm{bnd}}$, $s_{\mathrm{bnd}}^{2}$ in the monoid, where
\begin{align*}
\varphi_{\mathrm{int}}(\cdot) &= \sum_{q=1}^{n_{u}} \rho_{q}(\cdot) \big[{\lambda}_{\mathrm{int}}^{\top} A_{\mathrm{int}}(\x)\big]_{q} + \sum_{\ell} \rho_{\ell}^{\phi}(\cdot) \, \phi_{\ell}(\x) \\
&\quad + \rho^{s}(\cdot) \big(- {\lambda}_{\mathrm{int}}^{\top} {r}_{\mathrm{int}}(\x) - s_{\mathrm{int}}^{2}\big), \\
\varphi_{\mathrm{bnd}}(\cdot) &= \sum_{q=1}^{n_{u}} \rho_{q}(\cdot) \big[{\lambda}_{\mathrm{bnd}}^{\top} A_{\mathrm{bnd}}(\x)\big]_{q} + \sum_{\ell} \rho_{\ell}^{\phi}(\cdot) \, \phi_{\ell}(\x)\\
&\quad + \rho^{b}(\cdot) b(\x) + \rho^{s}(\cdot) \big(- {\lambda}_{\mathrm{bnd}}^{\top} {r}_{\mathrm{bnd}}(\x) - s_{\mathrm{bnd}}^{2}\big).
\end{align*}
\end{proposition}
\begin{proof}
By Lemma~\ref{Lemma:Farkas}, infeasibility of the interior and boundary dual systems~\eqref{eq:dae_dual_int}-\eqref{eq:dae_dual_bnd} is equivalent to emptiness of their semialgebraic representations. By Theorem~\ref{th:Positivstellensatz}, such emptiness is certified by the stated cone–monoid–ideal identities. Under Assumption~\ref{assumption:polynomial_dae}, all functions are polynomial; hence, these conditions reduce to SOS feasibility problems.
\end{proof}

\subsubsection{Verification for non-polynomial Systems} 
When the system contains non-polynomial functions, the verification can be performed by solving nonlinear programs for the interior and boundary conditions. The interior condition holds if the optimal value of
\begin{equation*}
\begin{array}{ll} 
\min_{\x, {\lambda}_{\mathrm{int}}} & {\lambda}_{\mathrm{int}}^{\top} {r}_{\mathrm{int}}(\x) \\
\mbox{s.t.} & {\lambda}_{\mathrm{int}}^{\top} A_{\mathrm{int}}(\x) = {0}^{\top} \\
& {\lambda}_{\mathrm{int}} \geq {0} \\
& \phi(\x) = {0} \\
& b(\x) > 0.
\end{array}
\end{equation*}
Similarly, the boundary condition holds if the optimal value of
\begin{equation*}
\begin{array}{ll} 
\min_{\x, {\lambda}_{\mathrm{bnd}}} & {\lambda}_{\mathrm{bnd}}^{\top} {r}_{\mathrm{bnd}}(\x) \\
\mbox{s.t.} & {\lambda}_{\mathrm{bnd}}^{\top} A_{\mathrm{bnd}}(\x) = {0}^{\top} \\
& {\lambda}_{\mathrm{bnd}} \geq {0} \\
& \phi(\x) = {0} \\
& b(\x) = 0.
\end{array}
\end{equation*}
is non-negative. If an optimal value is negative, the corresponding $\x$ is a feasibility counterexample, denoted $\x^{(f)}_{ce}$, where no safe and compatible control exists.

\begin{remark}
    As in Remark~\ref{remark:nlp_verification}, NLP-based feasibility verification requires global optimality to avoid false negatives, and a negative optimal value yields a concrete counterexample.
\end{remark}

\section{Experiments}
\label{sec:experiments}
We validate the proposed DAE-aware CBF framework on two DAE systems: a wind turbine system and a flexible-link manipulator, demonstrating safety enforcement while respecting algebraic constraints for \ijrncrevision{both DAE models}. The experiments are conducted in a simulation environment using Python and are performed on a machine with an Apple M2 Max chip and 32 GB RAM running macOS. For both systems, we construct DAE-aware CBFs using the methods outlined in Section~\ref{sec:dae_cbf} and \ijrncrevision{evaluate their correctness and feasibility} using SOS and the SMT-based approach from Section~\ref{section:verification}. We then implement a CBF-QP controller that incorporates the DAE-aware CBFs to enforce safety during system operation. The CBF-QP controller is implemented using Drake~\cite{drake} to compute the control inputs that ensure safety. The safety specifications are defined as state constraints that the system must not violate.
The code is available at~\cite{zhang2025daecbf_code}.

According to the theoretical results presented in Sections~\ref{sec:dae_cbf} and~\ref{section:verification}, we investigate the following research questions.
\begin{itemize}
\item $(\mr{Q}1)$ For the wind turbine system in Example~\ref{exmpl:DAE_unaware_example_1}, does the DAE-aware CBF from Theorem~\ref{thm:positive_invariance} ensure feasibility on the constraint manifold and prevent the safety violations observed in the DAE-unaware case?
\item $(\mr{Q}2)$ Do the candidate barriers used in the case studies satisfy the DAE-aware CBF validity requirements in Definition~\ref{def:dae_aware_cbf}, as \ijrncrevision{assessed using} the correctness/interior/boundary checks in Theorem~\ref{thm:dae_aware_cbf_verification}?
\item $(\mr{Q}3)$ \ijrncrevision{Can the framework be applied to the nonlinear manipulator model while preserving its algebraic constraint?}
\end{itemize}
To answer the posed research questions, we first apply the proposed forward-invariance and verification methods to the wind turbine index-1 DAE, thereby addressing $(\mr{Q}1)$--$(\mr{Q}2)$. \ijrncrevision{We then consider the index-1 flexible-link manipulator to address $(\mr{Q}3)$ and assess $(\mr{Q}2)$ for this nonlinear system.}

\subsection{Wind Turbine System}
We revisit the semi-explicit DAE system~\eqref{eq:wind_turbine} from Example~\ref{exmpl:DAE_unaware_example_1}. Based on the results showing safety violations under DAE-unaware CBFs, we now apply Theorem~\ref{thm:positive_invariance} to ensure safety on the constraint manifold.
Unlike the DAE-unaware case where the QP became infeasible, leading to safety violations, the DAE-aware CBF explicitly incorporates the algebraic constraints into the safety condition.

A degree-2 polynomial candidate CBF with relative degree $d=1$, for the wind turbine system, is given as follows:
\begin{align*}
b(\x) =& -1175.36 + 238.42 x_1 + 1491.68 x_2 + 238.42 x_3 \\
& - 263.78 x_1^2 + 472.45 x_1 x_2 - 477.45 x_1 x_3 \\
& - 145.97 x_2^2 - 263.78 x_2 x_3 + x_3^2.
\end{align*}

Note that synthesizing candidate DAE-aware CBFs is beyond the scope of this paper and is left for future work.  The above $b(\x)$ is a hand-designed polynomial barrier used for demonstrating the proposed framework.

Simulation results in Fig.~\ref{fig:synth_results_daecbf} demonstrate that the system trajectory remains feasible on the constraint manifold and the safety function $h(\x)$ remains non-negative.
Fig.~\ref{fig:synth_results_daecbf}a shows the safe trajectory evolving on the manifold, and Figs.~\ref{fig:synth_results_daecbf}b and~\ref{fig:synth_results_daecbf}c confirm that the barrier condition is satisfied throughout the simulation. The controller enforces the compatibility condition~\eqref{eq:invariance_projection} together with the barrier inequality~\eqref{eq:invariance_cbf}. Given that the wind turbine barrier function has relative degree one on $\mathcal{M}$, Corollary~\ref{lemma:dae_hocbf_index1} reduces to the index-1 DAE-aware CBF condition (Theorem~\ref{thm:positive_invariance}).

\begin{figure*}[!htbp]
    \centering
    \includegraphics[width=0.95\linewidth]{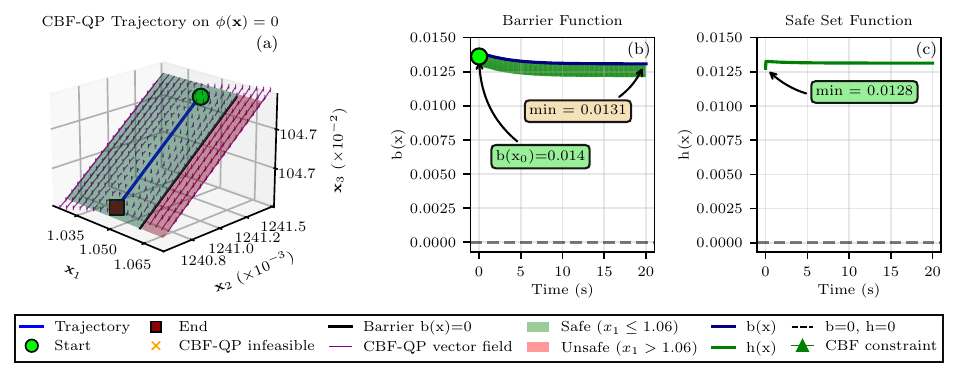}
    \caption{Simulation results for the DAE system~\eqref{eq:wind_turbine} using a DAE-aware CBF. The figure shows (a) the CBF-QP trajectory on the constraint manifold $\phi(\x)=0$ remaining in the safe region, (b) the evolution of the barrier function $b(\x)$ (with $b(\x_0)=0.014$ and a minimum $ b(\x(t))=0.0131$), and (c) the evolution of the safe set function $h(\x)$ (with a minimum $h(\x(t))=0.0128$). Both $b(\x)$ and $h(\x)$ remain positive over the horizon, indicating safety and feasibility.}
    \label{fig:synth_results_daecbf}
\end{figure*}

Furthermore, to ensure that the DAE-aware CBF is valid, we perform formal verification using both SOS programming and SMT-based methods, following the framework in Fig.~\ref{fig:roadmap_verification}. We verify the correctness, interior, and boundary conditions for the DAE-aware CBF based on Theorem~\ref{thm:dae_aware_cbf_verification}. The verification results along with runtimes are shown in Tab.~\ref{tab:windturbine_size_time} for the candidate CBF. The verification confirms that $b(\x)$ is a valid DAE-aware CBF and thus ensures safety on the constraint manifold.

By observing Fig.~\ref{fig:synth_results_daecbf} and Tab.~\ref{tab:windturbine_size_time}, we highlight that \textit{(i)} the DAE-aware CBF-QP remains feasible on $\mathcal{M}$ and maintains both $b(\x)$ and $h(\x)$ nonnegative throughout the trajectory, and \textit{(ii)} the correctness, interior, and boundary checks in Theorem~\ref{thm:dae_aware_cbf_verification} are computationally certifiable for this candidate barrier, thus answering questions $(\mr{Q}1)$ and $(\mr{Q}2)$.

\begin{table}[t]
    \centering
    \setlength{\belowcaptionskip}{\ijrnctablecaptionskip}
    \caption{Wind turbine: problem size and runtime. Each entry reports the number of parameters and runtime for verifying the corresponding property. Degree-4 SOS timing is omitted due to excessive cost (N/A).}
    \label{tab:windturbine_size_time}
    \small
    \renewcommand{\arraystretch}{1.3}
    \resizebox{\ifdim\width>\columnwidth\columnwidth\else\width\fi}{!}{%
    \begin{tabular}{l|cc|cc|cc}
        \midrule \hline
        \multirow{2}{*}{Property} 
        & \multicolumn{2}{c|}{SMT (Z3)} 
        & \multicolumn{2}{c|}{SOS ($d{=}2$)} 
        & \multicolumn{2}{c}{SOS ($d{=}4$)} \\
        \cline{2-7}
        & \#params & time (s) & \#params & time (s) & \#params & time (s) \\
        \hline
        Correctness & 3   & 0.01    & 60    & 0.04   & 280     & N/A \\
        Interior    & 7   & 120.65  & 405   & 5.77   & 8{,}235 & N/A \\
        Boundary    & 7   & 126.08  & 405   & 5.30   & 7{,}695 & N/A \\
        \hline
        Total       & 17  & 246.74  & 870   & 11.11  & 16{,}360 & N/A \\
        \toprule \bottomrule
    \end{tabular}}
\end{table}

\subsection{Flexible Link Manipulator}
We consider a two-link planar manipulator modeled as an \ijrncrevision{index-1} DAE following~\cite{moberg2007modeling,Park2024b}. The dynamics along with the algebraic constraint can be written as follows
\begin{subequations}\label{eq:manipulator_dae}
\begin{align}
 \begin{bmatrix}
    \dot{x}_1 \\ \dot{x}_2
\end{bmatrix} =& 
\begin{bmatrix}
    x_3 \\ x_4
\end{bmatrix}\label{eq:manip_kinematics}\\
 M(\x)
\begin{bmatrix}
    \dot{x}_3 \\ \dot{x}_4
\end{bmatrix} =& 
\begin{bmatrix}
    u_1 \\ u_2
\end{bmatrix} 
- {C}(\x) - {G}(\x) - {D}\begin{bmatrix}
    x_3 \\ x_4
\end{bmatrix}, \label{eq:manip_dynamics}\\
 0 =& \; \ell_1 \sin x_1 + \ell_2 \sin(x_1 + x_2) - x_5, \label{eq:manip_constraint}
\end{align}
\end{subequations}
where the differential states are $\x_d = [x_1, x_2, x_3, x_4]^{\top} \in \mathbb{R}^4$, such that $x_1, x_2$ represent joint angles ($\theta$) and $x_3, x_4$ represent joint velocities ($\omega$), with control inputs $\ci = [u_1, u_2]^{\top} \in \mathbb{R}^2$, representing joint torques with limit $\|\ci\|_\infty \leq 10$\,N$\cdot$m. Here, the algebraic constraint \eqref{eq:manip_constraint} enforces the end-effector vertical position, constraining the system to evolve on the constraint manifold $\mathcal{M} = \{\x \in \mathbb{R}^5 : \ell_1 \sin x_1 + \ell_2 \sin(x_1 + x_2) - x_5 = 0\}$. \ijrncrevision{The constraint~\eqref{eq:manip_constraint} directly determines the algebraic state $x_5$ from the joint angles.} The end-effector vertical position is defined as
\begin{equation}\label{eq:ytip_def}
y_{\mathrm{tip}} := \ell_1 \sin x_1 + \ell_2 \sin(x_1 + x_2) \quad \text{on } \mathcal{M}.
\end{equation}

The inertia matrix $M(\x)$, Coriolis vector $\boldsymbol{C}(\x)$, gravity vector $\boldsymbol{G}(\x)$, and damping matrix $\boldsymbol{D}$ are standard rigid-body dynamics terms. For brevity, we do not present these expressions; readers are referred to~\cite{moberg2007modeling} for detailed expressions. The parameters used in the simulation are: $\ell_1 {=} \ell_2 {=} 1$\,m, $m_1 {=} m_2 {=} 1$\,kg, $\ell_{c1} {=} \ell_{c2} {=} 0.5$\,m, $I_1 {=} I_2 {=} 0.083$\,kg$\cdot$m$^2$, $g = 9.81$\,m/s$^2$, $d_1 {=} d_2 {=} 0.1$\,N$\cdot$m$\cdot$s/rad.

\begin{table}[t]
    \centering
    \setlength{\belowcaptionskip}{\ijrnctablecaptionskip}
    \caption{Flexible link manipulator: problem size and runtime.
    Each entry reports the problem size and runtime for the corresponding verification subproblem. \ijrncrevision{\#SMT vars denotes the number of free real variables in each SMT verification query. The NN has $153$ stored parameters, with $4$ independent trained quantities. NN times cover the complete verification process; Total sums times only. \texttt{UNSAT} denotes a passed verification check.} }
    \label{tab:manipulator_size_time_vertical}
    \small
    \renewcommand{\arraystretch}{1.3}
    \resizebox{\ifdim\width>\columnwidth\columnwidth\else\width\fi}{!}{%
    \begin{tabular}{l|cc|cc}
        \midrule \hline
        \multirow{1}{*}{Item} 
        & \multicolumn{2}{c|}{Polynomial} 
        & \multicolumn{2}{c}{NN (Z3)} \\
        \cline{2-5}
        \hline

        Arch 
        & \multicolumn{2}{c|}{$b(\mathbf{x})=h(\mathbf{x})$}
        & \multicolumn{2}{c}{\ijrncrevision{8--8--tanh--8--tanh--1}} \\

        Verifier 
        & \multicolumn{2}{c|}{SOS ($d{=}2$)}
        & \multicolumn{2}{c}{\ijrncrevision{Z3}} \\

        \hline
        Subproblem   & \multicolumn{1}{c}{\#params} & \multicolumn{1}{c|}{time (s)} & \multicolumn{1}{c}{\ijrncrevision{\#SMT vars}} & \multicolumn{1}{c}{time (s)} \\
        \hline
        Correctness  & 165     & 0.09  & \ijrncrevision{3}  & \ijrncrevision{1.84 (\texttt{UNSAT})} \\
        Interior     & 2{,}565 & 4.90  & \ijrncrevision{4} & \ijrncrevision{1.77 (\texttt{UNSAT})} \\
        Boundary     & 2{,}565 & 4.48  & \ijrncrevision{5} & \ijrncrevision{1.74 (\texttt{UNSAT})} \\
        \hline
        Total        & \ijrncrevision{--} & 9.47  & \ijrncrevision{--} & \ijrncrevision{5.35} \\
        \toprule \bottomrule
    \end{tabular}}
\end{table}

\ijrncrevision{The displayed system~\eqref{eq:manipulator_dae} has differentiation index $\nu = 1$, since $\partial\phi/\partial x_5=-1$ and the algebraic equation directly determines $x_5$.} The safe set is defined as $\mathcal{C} = \{\x \in \mathcal{M} : h(\x) = y_{\max} - y_{\mathrm{tip}} \geq 0\}$ with $y_{\max} = 1.8$\,m, where $y_{\mathrm{tip}}$ is given by~\eqref{eq:ytip_def}. The safety function $h(\x)$ has relative degree $d' = 2$ on the constraint manifold (control appears in $\frac{d^2 h}{dt^2}$ but not in $\frac{dh}{dt}$). \ijrncrevision{The index-1 HOCBF construction is given by Corollary~\ref{lemma:dae_hocbf_index1}. We consider the barrier candidate $b(\x)=h(\x)$ with linear class-$\mathcal{K}$ functions of slope $\kappa=4$. This example does not validate the higher-index construction in Theorem~\ref{thm:positive_invariance_indexnu}.} The trigonometric nonlinearities make this system suitable for SMT-based verification from Section~\ref{section:verification}.

Simulation results are shown in Fig.~\ref{fig:motivating_infeasible_flexible}, where Fig.~\ref{fig:motivating_infeasible_flexible}a shows the trajectory in configuration space $(x_1, x_2, y_{\mathrm{tip}})$ remaining within the safe region, and Fig.~\ref{fig:motivating_infeasible_flexible}b shows the manipulator workspace trajectory under an aggressive initial swing ($\dot{x}_1 = 3$\,rad/s), where the end-effector (orange) stays below the $1.8$\,m limit (red dashed), reaching a maximum height of $1.57$\,m with a minimum safety margin of $0.23$\,m. Despite the high initial angular velocity, the DAE-aware CBF-QP controller successfully modulates the joint torques to keep the end-effector below the safety limit while respecting the algebraic constraint throughout the simulation.

\begin{figure*}[!htbp]
    \centering
    \includegraphics[width=0.95\linewidth]{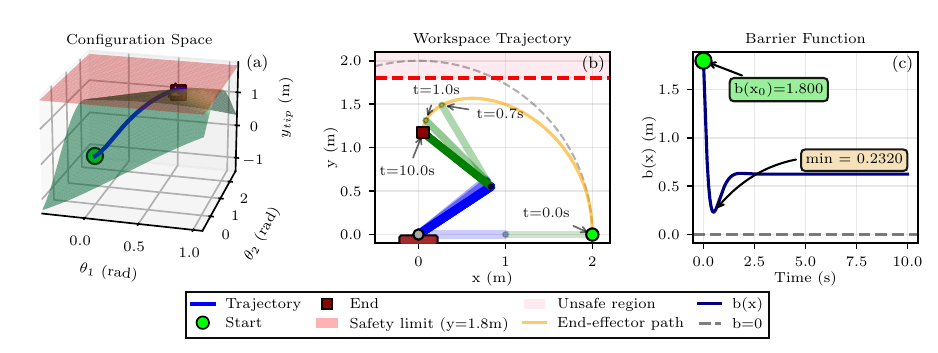}
    \caption{Simulation results for the flexible link manipulator. The figure shows (a) the trajectory in configuration space $(\theta_1,\theta_2,y_{\mathrm{tip}})$, where the trajectory (blue) starts at the green marker and ends at the red marker, remaining within the safe region ($b(\x)\ge 0$, green surface) and below the safety limit (red plane). (b) The manipulator trajectory in workspace, where the end-effector path (yellow) remains below the safety limit $y_{\mr{max}}=1.8\,\mathrm{m}$ (red dashed); the snapshots indicate the manipulator configuration at selected times. (c) Evolution of the barrier function along the trajectory (with $b(\x_0)=1.800$ and with a minimum $b(\x(t))=0.2320$), confirming that $b(\x(t))\ge 0 \, \forall \, t$.}
    \label{fig:motivating_infeasible_flexible}
\end{figure*}

To \ijrncrevision{assess the candidate DAE-aware CBF} for this case, we apply both polynomial (SOS) and neural network (SMT) verification from Section~\ref{section:verification}. The verification results along with runtimes are reported in Tab.~\ref{tab:manipulator_size_time_vertical}. For the polynomial case, we use $b(\x) = h(\x)$ directly and verify using degree-2 SOS, which completes in under $10$\,s. For the neural network case, \ijrncrevision{we use a new constrained candidate with a dense $8$--$8$--$\tanh$--$8$--$\tanh$--$1$ layout over explicit energy features, with $153$ stored positions holding exactly $4$ trained quantities. All three subproblems return \texttt{UNSAT} for their outer queries, which combined with the exact full-domain analytic proof establishes correctness, interior feasibility and boundary tangency over the declared domain. The candidate and the verification route differ from those used previously.} The table reports SOS parameter counts, SMT query variable counts, and verification runtimes. We note that this case shows that \textit{(i)} the DAE-aware CBF-QP controller can successfully enforce safety while respecting the algebraic constraint for \ijrncrevision{this index-1 nonlinear DAE system}, and \textit{(ii)} the same correctness/interior/boundary workflow of Theorem~\ref{thm:dae_aware_cbf_verification}, with SOS certificates and SMT-based falsification checks, can be applied to \ijrncrevision{assess the candidate barrier. The simulation addresses $(\mr{Q}3)$, and $(\mr{Q}2)$ is answered affirmatively for this new constrained neural candidate, whose correctness, interior and boundary conditions are all established over the declared domain.}

\section{Conclusion}
\label{sec:conclusion}
A framework for safety-critical control of DAE systems using CBFs is introduced. By incorporating projected vector fields defined on the constraint manifold, the proposed DAE-aware CBF conditions ensure forward invariance of safe sets while remaining consistent with the algebraic constraints. The framework is shown to extend to higher-index DAEs through high-order projected vector fields that account for hidden constraints and the system's differentiation index, giving rise to HOCBFs. Formal verification tools were introduced to certify geometric consistency and feasibility of candidate barrier functions. The proposed framework was validated on two DAE systems, demonstrating that DAE-aware CBFs can maintain safety without violating algebraic constraints, whereas standard CBFs may violate the algebraic constraints. 

The manuscript is \textit{not} devoid of limitations and thus merits the following future work. In particular, $(i)$ the SOS-based and NN-SMT-based verification methods suffer from scalability limitations, which restrict the applicability of the approach to high-dimensional DAE systems such as power systems; $(ii)$ the proposed framework assumes exact knowledge of the system model, and extending it to uncertain DAE systems remains an open problem; and $(iii)$ the current approach requires a candidate barrier function; thus, the problem of synthesizing DAE-aware CBFs from safety specifications is not addressed. This motivates future work on the synthesis of DAE-aware CBFs.

\section*{Statement of AI Use}
We used OpenAI Codex and Anthropic Claude Code for coding, debugging, and grammar correction. The authors reviewed the outputs and take responsibility for the final manuscript.

\section*{Data availability statement}
The project's public source-code repository is available at \url{https://github.com/HongchaoZhang-HZ/DAE-aware_CBF}.

\bibliographystyle{wileyNJD-Vancouver}
\bibliography{ref.bib}

\begin{IEEEbiography}[{\includegraphics[width=1in,height=1.25in,clip,keepaspectratio]{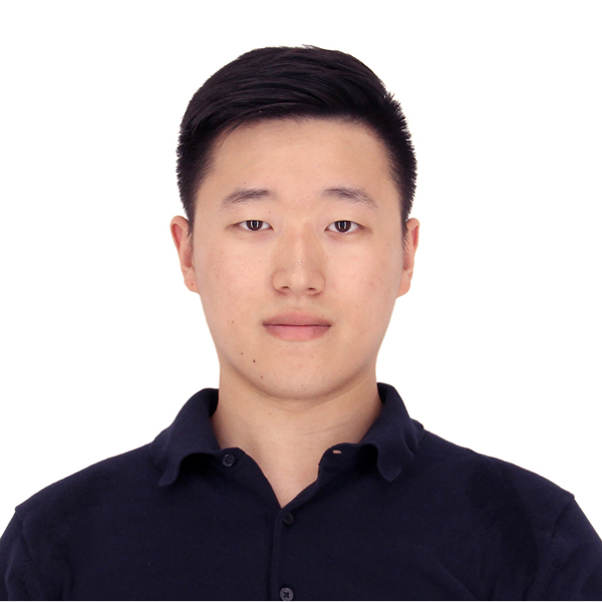}}]{Hongchao Zhang} is a Postdoctoral Scholar in the Department of Computer Science at Vanderbilt University. He received his Ph.D. degree in Electrical Engineering from the Department of Electrical and Systems Engineering at Washington University in St. Louis (WashU) in 2025, his M.Sc. degree in Electrical and Computer Engineering from the Department of Electrical and Computer Engineering, Worcester Polytechnic Institute (WPI) in 2020 and his B.Eng degree in Automation from the Department of Automation Engineering, Nanjing University of Aeronautics and Astronautics, Nanjing, China, in 2018.
He received the 2023 General Motors AutoDriving Security Award at the inaugural ISOC Symposium on Vehicle Security and Privacy at the Network and Distributed System Security Symposium (NDSS). His current research interests include learning-enabled control, Robotics, formal methods, and trustworthy AI. 
\end{IEEEbiography}

\begin{IEEEbiography}
	[{\raisebox{0\height}{\includegraphics[width=1in,height=1.25in,clip,keepaspectratio]{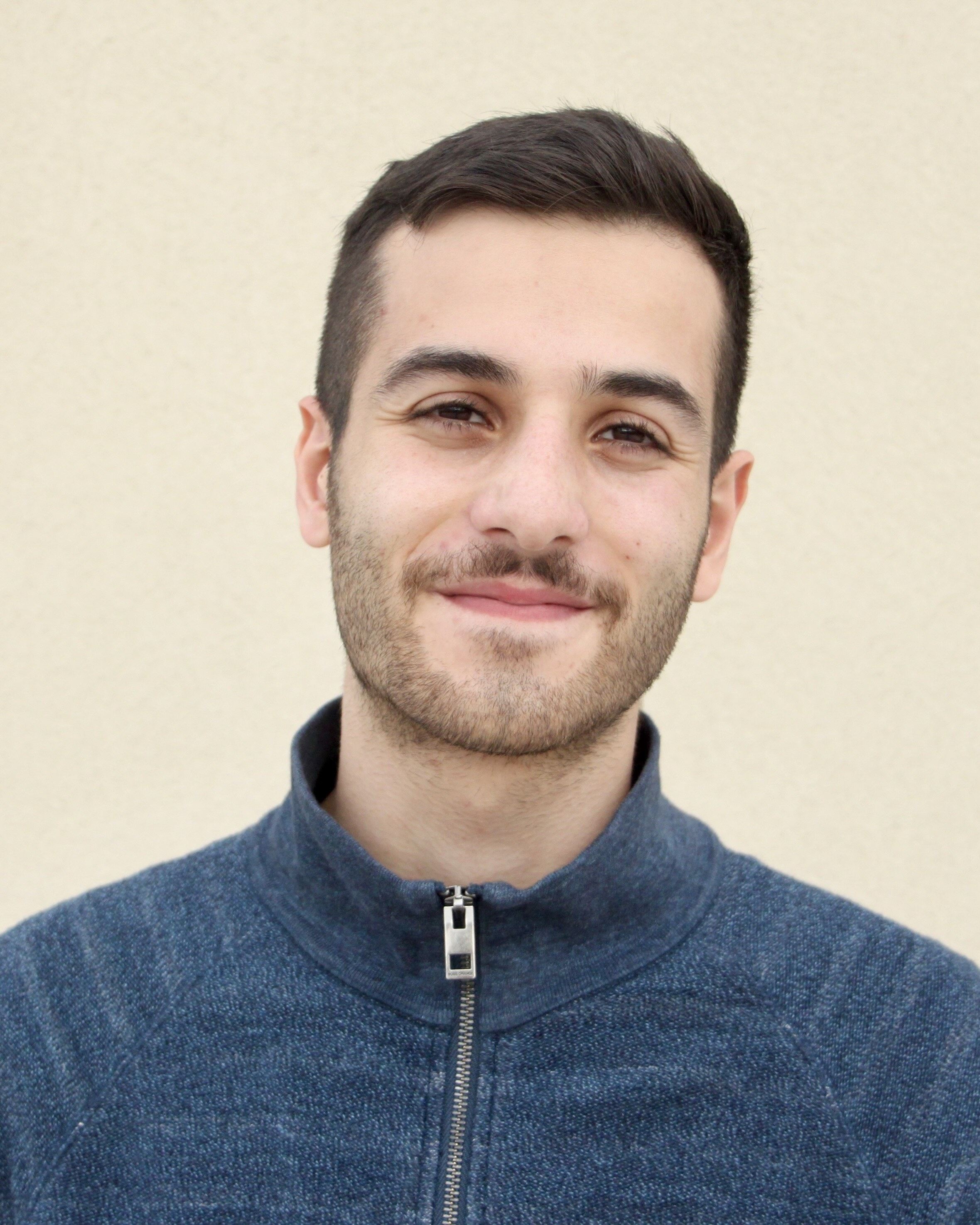}}}]
	{Mohamad H. Kazma} received the B.E. degree in civil and environmental engineering from the Notre Dame University-Louaize, Zouk Mosbeh, Lebanon, in 2019, and the M.Sc. degree in Civil Engineering from Vanderbilt University, Nashville, TN, USA, in 2024. He is currently a graduate research assistant and working toward the Ph.D. degree in civil engineering at Vanderbilt University, Nashville, TN, USA. His research interests focus on the modeling, optimization, and monitoring of nonlinear dynamical systems, with applications to complex infrastructure systems. His work leverages nonlinear control systems theory and advanced optimization techniques to address challenges in infrastructure networks.
\end{IEEEbiography}

\begin{IEEEbiography}
	[{\raisebox{0\height}{\includegraphics[width=1in,height=1.25in,clip,keepaspectratio]{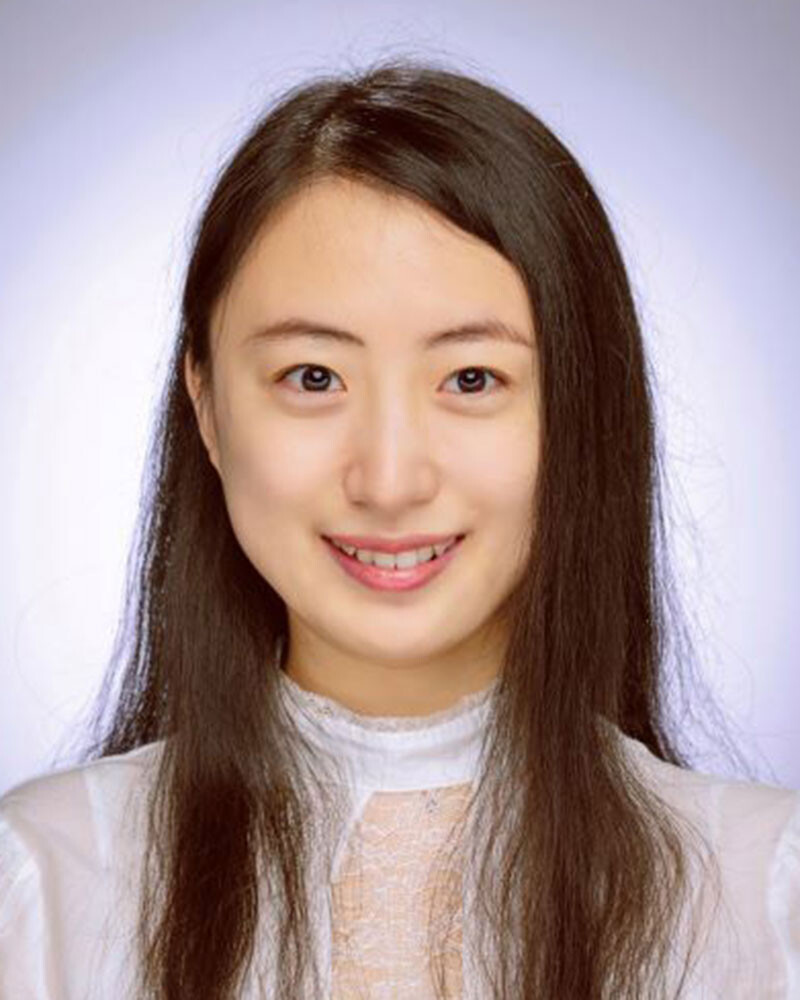}}}]
{Meiyi Ma} is an Assistant Professor in the Department of Computer Science at Vanderbilt University. She received her Ph.D. in Computer Science from the University of Virginia, Charlottesville, Virginia, USA., in 2021. Her research focuses on developing rigorous, explainable, and trustworthy AI methods by integrating formal methods with machine learning to enable safe, reliable, and robust cyber-physical systems in domains such as emergency response, public safety, smart cities, and healthcare. She has been recognized with awards, including the NSF CAREER Award, Google Academic Research Award, Best Paper and Best Artifact Awards at ACM/IEEE conferences, Rising Stars in EECS, and the VUSE Community Impact Research Award. Dr. Ma is an Associate Editor for ACM Transactions on Computing for Healthcare and Pervasive and Mobile Computing.   
\end{IEEEbiography}

\begin{IEEEbiography}
	[{\raisebox{0\height}{\includegraphics[width=1in,height=1.25in,clip,keepaspectratio]{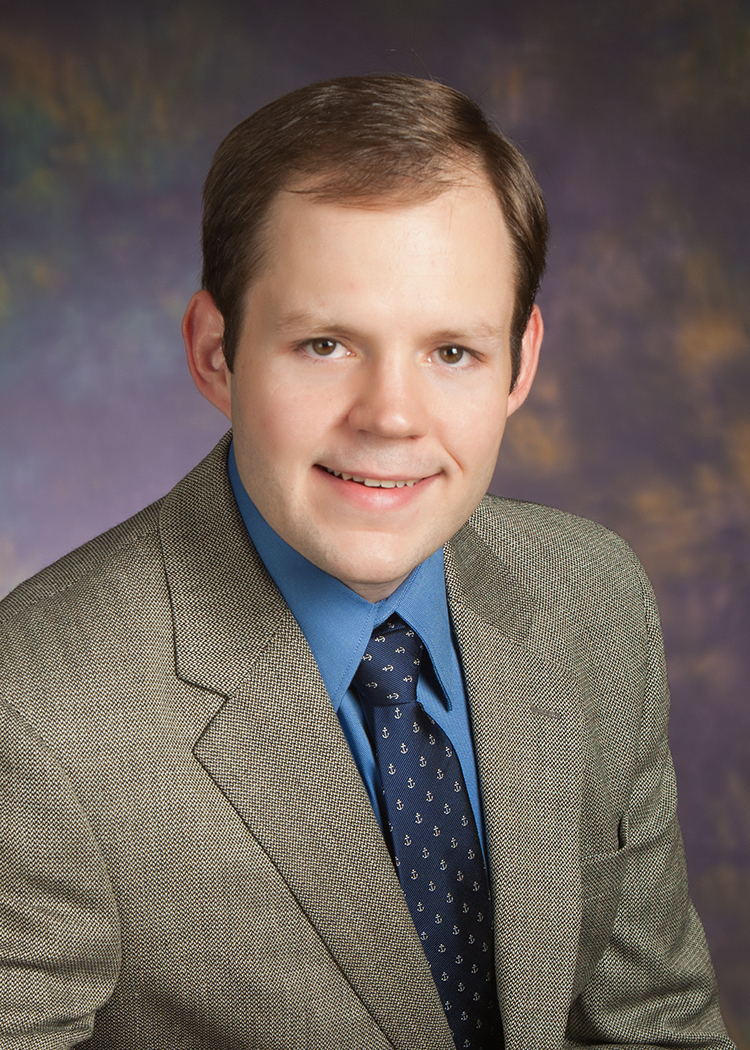}}}]
{Taylor T. Johnson} Taylor T. Johnson (Member, IEEE) received the B.S.E.E. degree in electrical and computer engineering from Rice University, Houston, TX, USA, in 2008, and the M.Sc. and Ph.D. degrees in electrical and computer engineering from the University of Illinois at Urbana-Champaign, Champaign, IL, USA, in 2010 and 2013, respectively. He is an Associate Professor of Computer Science with Vanderbilt University, Nashville, TN, USA. His research interests include formal methods and verification for autonomous cyber-physical systems (CPS) that incorporate machine learning components. Dr. Johnson is A. James and Alice B. Clark Foundation Faculty Fellow. He is a recipient of the Air Force Office of Scientific Research (AFOSR) Young Investigator Program (YIP) Award and several best paper and artifact awards.
\end{IEEEbiography}

\begin{IEEEbiography}
	[{\raisebox{0\height}{\includegraphics[width=1in,height=1.25in,clip,keepaspectratio]{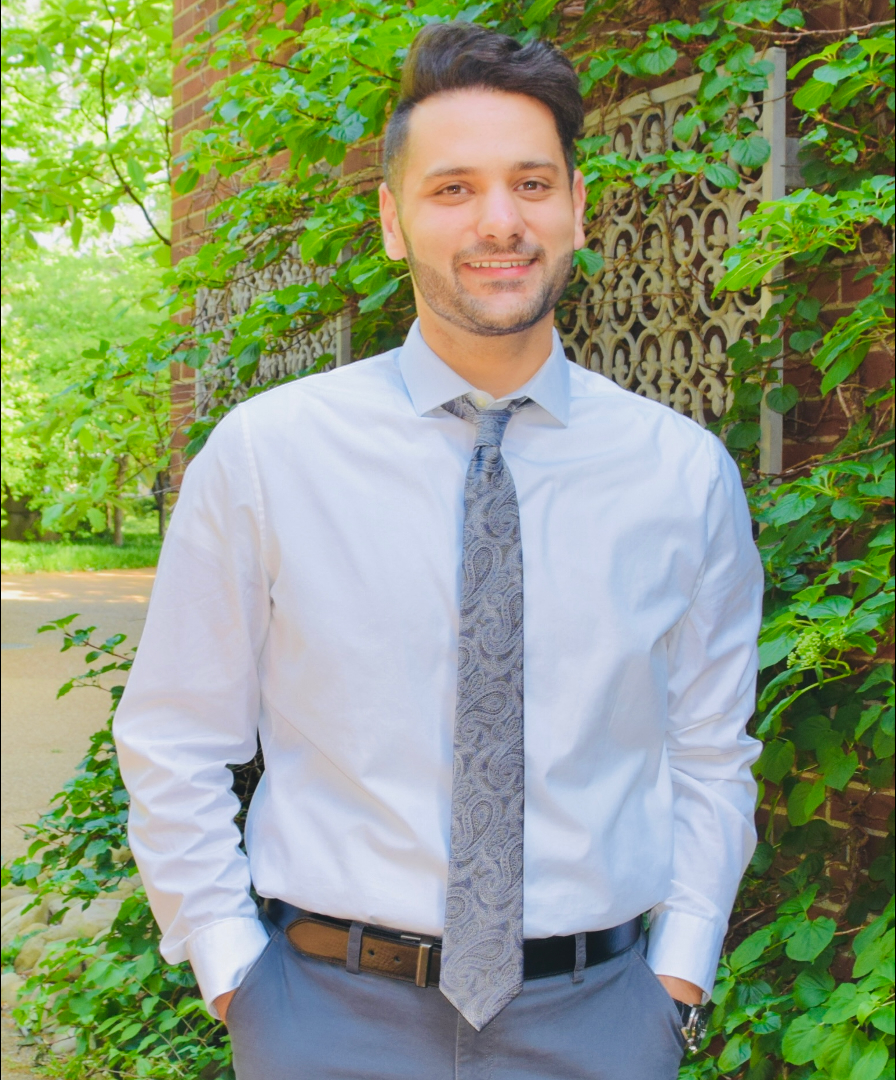}}}]
	{Ahmad F. Taha} received the B.E. degree in electrical and computer engineering from the American University of Beirut, Beirut, Lebanon, in 2011, and the Ph.D. degree in electrical and computer engineering from Purdue University, West Lafayette, IN, USA, in 2015. He is currently an Associate Professor with the Department of Civil and Environmental Engineering, Vanderbilt University, Nashville, TN, USA. He has a secondary appointment of electrical and computer engineering. Prior to joining Vanderbilt University, he was an Assistant Professor with the ECE department, The University of Texas, San Antonio, TX, USA. He is interested in understanding how complex cyber-physical, urban infrastructure operate, behave, and occasionally misbehave. His research interests include optimization, control, monitoring, and security of infrastructure with applications to power, water, and transportation systems. Dr. Taha is an Associate Editor for \textsc{IEEE Transactions on Control of Network Systems}.
\end{IEEEbiography}

\clearpage
\listoffigures
\end{document}